\documentclass[a4paper, english,12pt]{article}
\usepackage[T1]{fontenc}
\usepackage{babel}
\usepackage{textcomp}
\usepackage{amsmath}
\usepackage{setspace}
\usepackage{amssymb}
\usepackage{lscape}
\usepackage{caption}
\usepackage{subfig}
\usepackage{colortbl}
\usepackage{eucal}
\usepackage{amsfonts}
\usepackage{listings}
\usepackage[]{sidecap}
\usepackage[a4paper,top=3.8107cm,bottom=3.8107cm,left=2.5cm,right=2.5cm, bindingoffset=5mm]{geometry}
\usepackage{pdfpages} 
\pdfoutput=1
\makeatletter

\newtheorem{theorem}{Theorem}
\newtheorem{proof}{Proof.}

\newtheorem{proposition}{Proposition}

\makeatother
\begin{document}

\title{Can Market Risk Perception Drive Inefficient Prices? Theory and Evidence}
 \author{\vspace{0.5cm}\\\textbf{Matteo Formenti}\thanks{I thank Stefano Lovo (HEC, Paris) and Tommaso Proietti (Universit\'a di Roma Tor Vergata) for useful comments received during the first draft, and the participants to the European Financial Management Conference 2011. I am grateful to the anonymous referee of the Computing in Economics and Finance Conference. Last, I am in debt with Roberto Monte (Universit\'a di Roma Tor Vergata) who supervised me and encouraged me to study this field.}\\Group Risk Management\\UniCredit Group \\ \\ Universit\'a  Carlo Cattaneo}
\date{January 31, 2014}
\maketitle
\bigskip
\normalsize
\begin{abstract}

This work presents an asset pricing model that under rational expectation equilibrium perspective shows how, depending on risk aversion and noise volatility, a risky-asset has one equilibrium price that differs in term of efficiency: an informational efficient one (similar to Campbell and Kyle (1993)), and another one where price diverges from its informational efficient level. The former Pareto dominates (is dominated by) the latter in presence of low (high) market risk perception. The estimates of the model using S\&P 500 Index support the theoretical findings, and the estimated inefficient equilibrium price captures the higher risk premium and higher volatility observed during the Dot.com bubble 1995--2000. 
\end{abstract}
  \thispagestyle{empty}
  
\newpage{}

\pagenumbering{arabic}

\section*{Introduction}
This work presents an asset pricing model that under rational expectation equilibrium perspective shows how, depending on risk aversion and noise volatility, an economy has one equilibrium price that differs in term of efficiency: an informational efficient one (similar to Campbell and Kyle (1993)), and another one where price diverges from its informational efficient level. The combination of risk aversion and noise volatility, considered as exogenous to our economy, measure the ``risk perception'' of the informed investors who seek the optimal investment between a risky asset and a risk-free asset. In particular, the model shows that price is efficient when informed investors have a \textit{low} market risk perception, and is inefficient when investors' risk perception is \textit{high}, and it proves that the efficient one Pareto dominates (is dominated by) the inefficient one in presence of low (high) market risk perception.

The work uses the economy described in Campbell and Kyle (1993) (hereinafter C.K.) which considers a riskless and a risky asset, and investors  trading with noise traders. The latter aim to capture the irrationality of financial market, such that price is written as a linear function of the investors' information set and a noise variable. It is important to remark that, in their model, price \textit{is always ``semi-strong'' efficient} because it equals the sum of the expected future discounted dividend (Fama (1970)) and the noise term. Finally, they add a constant term to capture the risk-premium demanded by a risk-averse investors and they find evidence of such ``semi-strong'' efficient price form using real data (S\&P 500 Index).

This work extends the celebrated model \textit{assuming that price is not always} the expected future discounted dividend although investors still assume a price in the ``semi-strong'' form. The theoretical model, and the empirical evidence, finds an efficient price only when risk-aversion and noise volatility are \textit{low}, while it deviates from such efficiency when risk perception is \textit{high}. Result is obtained in the theoretical model: (i) letting the coefficients of the ``semi-strong'' price be misspecified, when investors search the optimal investment strategy; (ii) searching with simulations the price that makes optimal the demand, for a given set of exogenous parameters; (iii) using the maximum Utility criteria to choose among several candidates equilibrium prices, each solution of the investment problem, the equilibrium price.\footnote{The Pareto-dominance is proved numerically due to the high degree of non linearity of equations in the optimal investment problem.}

Put in other words, in C.K. price is \textit{a-priori efficient} and the investors find the corresponding Utility value once the exogenous parameters of the economy are estimated. In this work price is \textit{a-priori inefficient} thus depending on the level of the exogenous variables either the informational efficient or the informational inefficient equilibrium will have the highest Utility. 
Such equilibrium price can be the efficient one, when risk aversion and noise volatility is \textit{low}, and inefficient when risk perception is \textit{high}. 

The following implications of the main result derive for financial markets: (i) investors consider an inefficient price as a profitable condition if risk perception is \textit{high}; (ii) a combination of high risk aversion and high price volatility solely lead the result, that is a \textit{low} (\textit{high}) risk aversion combined with \textit{high} (\textit{low}) noise volatility does not necessarily allow to determine an inefficient price; (iii) when risk aversion is high, there exists a threshold value of noise volatility, and a consequent threshold value of price volatility, inducing the investors to have a higher Utility with inefficient equilibrium price. 

The empirical estimates of the model support the theoretical results. Using the long time series of S\&P 500 Index (1871-2009) and Nasdaq Index (1974-2009) there are the following results: (i) data reflects the fundamental value, that is the efficient price, during the long period 1871-2009; (ii) data shows that S\&P 500 price Index was not efficient during the period 1995-2000 known as the Dot.com bubble; (iii) Nasdaq price Index was not efficient during the period 1995-2000. The likelihood ratio test, used to compare the fit of the two models, rejects at 1\% the null hypothesis of the efficient price. In conclusion, the estimates seem to confirm the theoretical results that focus on the role of market noise volatility and investors' risk aversion to determine that inefficient price.

Furthermore, the empirical estimates shows that (i) investors ask a higher risk premium when risk perception is high, and the inefficient price holds in the market; (ii) there is an high estimated value of market volatility during the Dot.com bubble. The two evidences confirm the well-known empirical anomalies such as the equity premium and the excess volatility puzzle, as a key-drivers of the higher risk involved in period of market turbulence.

Last, it is important to remark the limits of this model: (i) it does not help to explain the endogeneity relationship between risk perception and efficient (inefficient) equilibrium price, because the risk perception's parameters are treated as exogenous; (ii) it does not analyze the dynamics from the efficient to the inefficient equilibrium (and conversely), inasmuch as any equilibrium price is solution of the investors investment optimization problem.
\bigskip

\section{Literature}
The idea that financial markets are perfect and correctly report the information has been widely debated since Fama (1970)\footnote{\, The term was originally coined in an unpublished working paper by Harry Roberts (1967), whereas the history of the efficient market hypothesis begins with Cardano in 1564 as reported by Sewell (2008).} introduced the Efficient Market Hypothesis (EMH) as an equilibrium condition to test in asset pricing models. In turn, what has not been adequately analyzed thoroughly is that, undesirable as they may be, inefficient prices characterize financial markets and, according to this work, they can be an equilibrium phenomenon desired by the market participants. In efficient markets prices are informative, transmit and make public the private information of the investors. Conversely, a market where price changes are not entirely due to the arrival of new information is inefficient because prices do not report correctly all the available information.

The works of Grossman (1976), Grossman and Stiglitz (1976), and later Kyle (1985, 1989) and C.K. (1993) are the first to include EMH and the presence of informed investors and noise traders in a contest of competitive markets.\footnote{\, See Brunnermeier, 2001 for an extended literature review.} These works study the role of information in price dynamics considering: (i) the informed investors as rational active traders who know everything and their trading is perfect, and (ii) noise traders as such investors who do not collect information and whose trading activity is informative for others. In such models, price can deviate from the fundamental value due to the action of noise traders and the desire of the rational investors to exploit them as much as possible. REE perspective had great success in capturing the dynamics and informativeness of asset prices, mainly because of the easy tractability of the equilibrium price and, because of the existence of a semi-strong efficient price. In such models noise traders play a key role in clearing the market and avoiding breakdown of the market.\footnote{\, Without the presence of noise traders, the no-trade theorem of Milgrom and Stokey (1982) applies and there is no exchange.} Despite these important theoretical results, and some critiques for the induced 'schizophrenia'' of the informed investors (Hellwig (1980), Kyle (1989), Back (1992)), some important market empirical anomalies have not been captured (see Siegel (2002) for an extensive review of all anomalies). The financial literature of last twenty years reviewed these apparent anomalies, taking the efficient markets hypothesis as a benchmark. They include the equity premium puzzle (Mehra and Prescott (1985)), the excess volatility in stock returns and price-dividend ratios (Grossman and Shiller (1981), LeRoy and Porter (1981), Shiller (1981)), and the predictability of stock returns (Poterba and Summers (1988), Fama and French (1989), Campbell and Shiller (1988)). According to Shiller (1998) these anomalies suggest that the underlying principles of rational behavior, and the efficient markets hypothesis, are not entirely correct and that we need to look 'as well at other models of human behavior'. This model shows that the informational efficient and inefficient price are both an equilibrium condition when both rational and irrational investors trade in the market.

The theoretical findings of this work are closer to Monte et al. (2010) who show, in an asymmetric information setting similar to Wang (1993), that financial markets admit different equilibria: a \textit{low} investors' risk perception induces investors to trade as perfect competitors, and consequently informationally efficient equilibria are achieved, while a \textit{high} investors' risk perception leads investors to behave as imperfect competitors, and informationally inefficient equilibria result. This work does not consider the asymmetry in the information structure, such as difference among informed or uninformed investors, though it is the first to support empirically the main theoretical findings.  

The model has been recently extended (see Monte and Formenti, w.p. 2012) considering two investors who, besides observing the publicly known dividend realizations, hold two different pieces of private information on the growth rate of dividend. The two groups differently informed compete each other to rationally extract their missing pieces of information from the demand for the risky asset of their competitors. In a Bayesian Nash equilibrium perspective, each group ends up with forecasting the forecasts of their competitors. A change in the assumption do not modify the results shown in this work. Besides the full informative equilibria, in which all private information is revealed and the competition of the two groups of investors ceases, it is shown that partially informative equilibria exist and the learning process, as well as the competition between the two groups, is never ending. As a major consequence, in partially informative equilibria the risky asset price reflects inefficiently the private information. Moreover, there is still evidence that partially informative equilibria Pareto dominate the full informative ones on the increasing of noise volatility and the investors' risk aversion.

\medskip

The paper is organized as follows. The economy is spelled out in Section \ref{Economy}.  Section \ref{Equilibrium} shows the model equilibria, the efficient price (Equilibrium-Type A) and the inefficient price ones (Equilibrium-Type B), and the utility criteria used to compare them. Section \ref{Numerical-Solution} calibrates the model and it shows the theoretical results. Section \ref{estimation} provides the estimates of the model and it shows that real data supports the main theoretical results. Section \ref{Conclusion} concludes.

\section{The Economy}
\label{Economy}

Consider an economy composed of informed risk-averse rational investors and noise traders exchanging a risky asset. The informed investors, or large trader interchangeably, are risk averse and observe the dividend process having a private information regarding the dividend growth rate; the noise traders represent the economic agents who exchange the asset without maximizing any Utility function, appear in the market randomly to buy/sell the asset for liquidity reasons and capture the irrationality of the market. Assume that dividend and price are normally distributed and that changes in the \emph{level} of dividend and stock price have constant variance. As a consequence the variance of percentage returns and dividend growth rate increases (decreases) when the level of price and dividends decreases (increases). Now de-trend dividend and stock prices by an exponential growth trend $\xi$, obtained by market data as the dividend growth mean. The de-trending operation on price and dividend let the new variables follow an Ornstein-Uhlenbeck process. According to C.K. (1993) the de-trended operation has two other effects: it removes exponential growth from the ex-ante mean of the data and from the variance of data. The latter one has an effect similar to a log transformation such that 
\begin{equation}
\label{D_U-P_U}
D(t)\equiv D^{u}(t)e^{-\xi t} \qquad \qquad P(t)\equiv P^{u}(t)e^{-\xi t}
\end{equation}
where the variables $D^{u}$ and $P^{u}$ are the observed dividend and price of the stock, and $D(t)$ and $P(t)$ are the de-trended dividend and price. Assume that changes in $D_{t}$ and $P_{t}$ are homoskedastic and normally distributed and, both have one unit root with a particular combination of levels of price and dividend that is stationary, that is they are cointegrated. Let $r$ be the time-invariant riskless interest rate. A permanent one-dollar changes in the de-trended dividend has a discounted value of $1/(r-\xi)$ dollars, and the cointegrating stationary vector is $D(t) - (r-\xi) P(t)$ with the unconditional mean equals to $\gamma$ = $E[ D(t) -(r-\xi)P(t)]$. In turn, the investors decompose price into the sum of a fundamental value, a constant term, and a noise term:
\begin{align} 
\label{Price-Equation}
P(t)=&p_{0}+V(t)+\Theta(t)
\end{align} 
where $p_{0}\equiv \frac{\gamma}{r-\xi}$ captures the constant risk premium per share of stock demanded by risk-averse informed investors, $V(t)$ represents the expected future dividend (i.e., public information) and the non-dividend component (interpreted as the investors' private information):
\begin{equation}
\label{V}
V(t)=p_{D_0}D_{0}(t)+p_{D_1}D_{1}(t)+p_{I}I(t)
\end{equation}
and $\Theta(t)$ is the noise component. This price's form is convenient because (i) it is linear, (ii) it does not require any assumption about the discount rate, and (iii) the noise trading component, given by a random supply of the stock, captures the presence of liquidity traders. Finally, note that noise trading influences the stock price because the informed investors are risk-averse and ask a risk premium that is captured by the constant term. A special case of Equation (\ref{Price-Equation}) is the benchmark case (Equilibrium-Type A) and the efficient price derived in the work of C.K. (1993) in which the fundamental value $V(t)$ is the present expected value of dividends and non-dividend discounted at the risk-less rate $r$.

\paragraph*{Dividend structure.}

There are continuous dividend announcements to the market. The de-trended dividend is the sum of the permanent and temporary components, independently distributed, both of which are directly observed by the informed investors:
\[
D(t)=D_{0}(t)+D_{1}(t).
\]
The permanent component is a brownian motion process, and the temporary component is a mean reverting process, a continuous-time
AR(1), given by
\begin{equation}
\label{Dividend-equ.}
dD_{0}(t) =\alpha_{I}I(t)dt+\sigma_{0}dw_{0}(t), \qquad 
dD_{1}(t) =-\alpha_{D}D_{1}(t)dt+\sigma_{D}dw_{D}(t)
\end{equation}
where $dw_{0}(t)$ and $dw_{D}(t)$ are two standard independent brownian motions, $\sigma_{0}$ and $\sigma_{D}$ constitute the innovations in $D_{0}(t)$ and $D_{1}(t),$ and the quantities $\sigma_{0}^{2}$ and $\sigma_{D}^{2}$ are the innovation variance of $D_{0}(t)$ and $D_{1}(t)$ respectively. The idea that dividends have a private hidden information content is an old one (see Lintner (1956), Miller and Modigliani (1961), Watts (1973)) and it has been tested empirically by several works (Shiller (1981), DeAngelo \textit{et al.} (1992)). The parameter $\alpha_{I}$ captures this hidden private information content in the dividend process and is useful for scaling the unit of $I$. $I(t)$ measures how much $D_{0}(t)$ is expected to increase in the future. The positive parameter $\alpha_{D}$ measures the mean speed reversion of the transitory component and $-\alpha_{D}D_{1}(t)$ measures the expected growth rate of dividend.\footnote{\,  The scaling parameter does not change the final results (see C.K. (1993), Appendix A).}

\paragraph*{Information structure.}

Informed investors receive private information $I(t)$ (a private signal) about the asset price. It is convenient to interpret $I(t)$ as the ``non-dividend information'' component. C.K. (1993) defines $I(t)$ as the measurement error on the transitory component. The information process is defined as $I(t)\equiv \hat{D}_{1}(t)-D_{1}(t)= D_{0}(t)-\hat{D}_{0}(t)$, where $[\hat{D}_{0}(t); \hat{D}_{1}(t)]$ are the investors' estimates of $D_{0}(t)$ and $D_{1}(t)$ respectively. The information dynamics is a mean reverting process
\begin{equation}
\label{Info-equ.}
dI(t)=-\alpha_{I}I(t)dt+\rho_{I}\sigma_{0\,}dw_{0}(t)+(2\rho_{I}-\rho_{I}^{2})^{1/2}\sigma_{0\,}dw_{I}(t)
\end{equation}
in which $dw_{I}(t)$ is a standard brownian motion independent of $dw_{0}(t)$, and $\sigma_{I}$ constitutes the innovation in $I(t)$. Investors receive new information about the traded stock captured by the two random components of the process, measured by the standard deviations $\sigma_{0}$ and $\sigma_{I}$. The parameter $\alpha_{I}$ captures the mean-reverting speed at which the new information is updated into the price. As $\alpha_{I}$ increases (decreases), private information decays faster (slower) and it is short lived in the price dynamics. The correlation structure between $dD(t)$ and $dI(t)$ is given by $\chi=-\frac{\rho_{I}}{\sqrt{2\rho_{I}}}$, which ensures that 
\begin{equation}
\label{div-info}
E\left\{ I(t+s)\,\mid D[-\infty,t]\right\} =0, \qquad \qquad s\geq0
\end{equation}
and the history of the dividend process cannot forecast the future of $I(t)$.  A technical condition $0\leq \rho_{I}\equiv\sigma_{I}^{2}/2\sigma_{0}^{2}\leq 2$ guarantees that $D$ does not forecast $I$, and the variables $D(t)$ and $I(t)$ are independently distributed. C.K. (1993) show that Equation (\ref{div-info}) uniquely determines the diffusion term in the $[D(t), I(t)]$ processes.
  
\paragraph*{Noise trading.}

Following noisy rational expectational models, the total amount of risky asset supply is $1+\Theta(t)$. The process $\Theta(t)$ models the deviation of the current risky asset supply from its long-run stationary level normalized to 1 and it implies that noise traders have inelastic demand of $1-\Theta(t)$ shares of the stock at time $t$. Moreover, $\Theta(t)$ is interpreted as the number of remaining shares available to the market. The noise process has a non-null mean reverting dynamics
\begin{equation}
\label{Theta-equ.}
d\Theta(t)=-\alpha_{\Theta}\Theta(t)\,dt+\sigma_{\Theta\,}dw_{\Theta}(t),
\end{equation}
in which $dw_{\Theta}(t)$ is a standard Brownian motion independent by $[dw_{D_{0}}(t)$, $dw_{D_{1}}(t)$, $dw_{I}(t)]$, the positive parameter
$\alpha_{\Theta}$ is the constant mean speed of reversion of the process $\Theta(t)$ towards its long-run null level, and $\sigma_{\Theta}$ is the volatility of noise. The stochastic supply of the risky asset in the aggregate market makes the market incomplete.

The informed investors observe the history of $D(t), I(t)$ and $P(t)$ so at time $t$ their information set is
\[
 \mathfrak{F}(t)\,\equiv \, \sigma [D_{0}(s),D_{1}(s),I(t),\Theta(t),P(t); \, s\leq t]\,=\, \sigma [D(t),I(t),P(t); \, s\leq t].
\]
The informed investors direct observation of $I(t)$ and $\Theta(t)$ implies that in equilibrium the observation of the price (\ref{Price-Equation}) is equivalent to the observation of the signal $p_{D_{0}}D_{0}(t)+p_{D_{1}}D_{1}(t)$. On the other hand, the observation of the public dividend $D(t)$ is equivalent to the observation of the signal $D_{0}(t) + D_{1}(t)$. Therefore, from the observation of $D(t), I(t)$ the investors can also observe $D_{0}(t)$ and $D_{1}(t)$.

\paragraph*{CARA-Utility.}
The informed investors have a constant absolute-risk aversion (CARA) Utility function
\begin{equation}
\label{cara-utility}
u[t, c(t)]=-e^{-[\beta t+\varphi c(t)]},
\end{equation}
where $\beta$ is the time-impatience parameter and $\varphi$ is the coefficient of the absolute risk aversion. The use of the CARA Utility function, and the assumption of normality of dividend and stock prices, let the expected future dividend be discounted at the riskless rate of interest. This is equivalent to saying that an increase in the expected future dividend, given by a higher value of investors' private information, is captured by a change in the variable $V(t)$ in Equation (\ref{Price-Equation}), while an increase (decrease) in the informed investors' risk aversion causes a lower (higher) risk premium captured by the higher (lower) constant term. Investors choose consumption and inventory of risky assets to maximize their Utility given the information set 
\[
\max_{\Psi(t),\, c(t)} \mathbf{E} \left[- \int_{t=0}^{+\infty}u[t, c(t)]\,dt\,|\mathfrak{F}(t)\right]
\]
The use of CARA preferences implies the investors' optimal asset demand and the optimal equilibrium price are independent of their wealth distribution as well as the level of aggregate wealth. This is why CARA Utility greatly simplifies the optimization problem. The model has a closed-form solution according to CARA preference. 

\section{Equilibrium}
\label{Equilibrium}
The equilibrium of the economy described in Section \ref{Economy} is solved using the REE perspective developed by Lucas (1972), Green (1973), Grossman (1976), and Kreps (1977).\footnote{\, The book by Hens and Schenk-Hopp (2009) presents a wide discussion of different theoretical, analytical, and empirical techniques under REE that explain the market dynamics of asset prices.} The mechanism is the following one: (i) investors solve the optimization problem and maximize the Utility considering the market-clearing price as parametric; (ii) investors and noise traders submit a scheduled demand to a Walrasian auctioneer; (iii) the auctioneer announces a price and receives from all market participants what their demand/supply would be at that price; (iv) the auctioneer clears the market and determines the equilibrium price. 

Following the procedure described above, the investors conjecture the price's form linearly depending on the state variables of the economy:
\begin{align}
\label{Price-Conjectured}
P(t)=\:&p_{0}+p_{D_0}D_{0}(t)+p_{D_1}D_{1}(t)+p_{I}I(t)+\Theta(t)
\end{align}
in which $p_{0}$ is expected to be negative, as it is the discount on price requested by the risk averse investors due to risk-premium, the state variables $[D_{0}(t), D_{1}(t)]$ account for the observed dividend $D(t)$, $I(t)$ is the hidden stationary private information, and $\Theta(t)$ is the aggregate supply shock of the stock. The variance of the stock price is given by $\sigma^{2}_{P}= p_{D_0}^{2}\sigma^{2}_{D_{0}}+p_{D_1}^{2} \sigma^{2}_{D_{1}}+p_{I}^{2} \sigma_{I}^{2}+\sigma^{2}_{\Theta}$. Assume a constant variance of price implies the variance of percentage of returns increases as the price of the stock decreases and vice versa. This is a phenomenon studied since the works of Black (1976) and Nelson (1987).

\paragraph*{Investment opportunity.}

Stock price has the following process
\begin{equation}
\label{dP}
dP(t) = [-p_{D_{1}}\alpha_{D} D_{1}(t)+(p_{D_{0}}\alpha_{I}-p_{I}\alpha_{I})I(t)-\alpha_{\Theta}\Theta(t)]\,dt + H \,dw(t)
\end{equation}
where $H = \{p_{D_{0}}\sigma_{0}+p_{I} \rho_{I} \, \sigma_{0}, \, \, p_{D_{1}}\sigma_{D}, \, \, p_{I}\sqrt{2\rho_{I}-\rho_{I}^{2}} \sigma_{0}, \, \, \sigma_{\Theta}\}$ and $dw(t)$ is the vector of brownian motions. The investment opportunity given by $Q(t)$ is the instantaneous excess return on one share of risky asset, and it is given by the process
\begin{equation}
dQ(t) = [D(t)-rP(t)]\,dt+dP(t)
\end{equation}
where the risk-less rate $r$ is assumed to be constant. $Q$ is the undiscounted cash flow from the zero-wealth portfolio while $dQ(t)$ is interpreted as the return on a zero-wealth portfolio long of one share of stock fully financed by borrowing at the risk-less rate (see Wang (1993)). 

Given the process (\ref{dP}), $Q(t)$ satisfies the stochastic differential equation
\begin{align}
\label{dQ-return}
dQ(t) =&  \:[D(t)-rP(t)]\,dt +dP(t)  \\
 =& \:[e_{0}+e_{D_{0}}D_{0}(t)+e_{D_{1}}D_{1}(t)+e_{I}I(t)+e_{\Theta}\Theta(t)]\,dt+ H \,dw(t)\nonumber
\end{align}
where $e_{0}=-r p_{0}$, \, $e_{D_{0}}=1-r p_{D_{0}}$, \,  $e_{D_{1}}=1- p_{D_{1}}(r-\alpha_{D})$, and $e_{\Theta}=-(r+\alpha_{\Theta})$. The conditional expectation of the excess return of one share of stock is $E[dQ] = [ e_{0}+e_{D_{0}}D_{0}+e_{D_{1}}D_{1}+e_{I}I+e_{\Theta}\Theta]\,dt$.
Notice that the expected excess return is affected by all the state variables of the economy while, noise volatility directly influences price volatility without affecting the investment opportunity. As in Wang (1993), the level of aggregate stock supply affects $dQ$ because it determines the total risk exposure of the economy.

\paragraph*{The optimization problem.}
The investors' wealth $W(t)$ has the following dynamics:
\begin{equation}
\label{dW-wealth-equ.}
dW(t)=[rW(t)-c(t)]\, dt+\Psi(t)\,dQ(t),
\end{equation}
where $c(t)$ is the investor's consumption policy, $\Psi(t)$ is the investors' inventory as the holding of the risky asset at time $t$. Investors maximize the expected value of the exponential Utility over the infinite time horizon, subject to the wealth dynamics and given the information set at time $t$, by controlling their inventory $\Psi(t)$ and their consumption $c(t)$. The investor's optimization problem is 
\begin{align}
\label{Investors'-objective-function}
\max_{\Psi(t),\, c(t)} \mathbf{E} & \left[- \int_{t=0}^{+\infty}e^{-[\beta t+\varphi c(t)]}\,dt\,|\mathfrak{F}(t)\right] \\ \nonumber
 s.t. \quad dW(t)=\:&[rW(t)-c(t)]\, dt+\Psi(t)\,dQ(t)
\end{align}
where $\mathbf{E}[\,\cdot\,|\mathfrak{F}(t)]$ is the conditional expectation operator given the information set $\mathfrak{F}(t)$. 
Let $J(Z, W, t)$\footnote{\, The assumption is that $J(Z,W,t)$ is twice differentiable in each of the state variables.} be the value function, where $(Z, W)$ are the state variables moving the investment opportunities and $Z=\left(1,\, D_{0},\, D_{1},\, I,\,\Theta \right)^{\top}$. The variables of the economy can be written in compact form as a Vector Autoregression (VAR).
\begin{equation}
\label{compact-form}
dZ(t)=AZ(t)\,dt+B^{1/2} dw(t)
\end{equation}
where\begin{align}
\label{A-B-Equation}
A & \equiv\left(\begin{array}{ccccc}
0 & 0 & 0 & 0 & 0\\
0 & 0 & 0& \alpha_{I}  & 0\\
0 & 0 & -\alpha_{D} & 0 & 0\\
0 & 0 & 0 & -\alpha_{I} & 0\\
0 & 0 & 0 & 0 & -\alpha_{\Theta}
\end{array}\right),
\quad 
B^{1/2}\equiv\left(\begin{array}{cccc}
0 & 0 & 0 & 0\\
\sigma_{0} & 0 & 0 & 0\\
0 & \sigma_{D}& 0 & 0\\
-\rho_{I}\sigma_{0} & 0& (2\rho_{I}-\rho_{I}^{2})^{1/2}\sigma_{0} & 0\\
0 & 0 & 0 & \sigma_{\Theta}
\end{array}\right).
\end{align}
The value function $J(Z, W, t)$ satisfies the Bellman equation

\begin{align}
\label{bellman-equ}
0=&\max_{\Psi(t),\, c(t)}\left\{ -\int_{t=0}^{+\infty}e^{-[\beta t+\varphi c(t)]}+E\left[dJ(Z, W, t)\right]ds\right\} \\ \nonumber
dZ(t) =& AZ(t)\,dt+B^{1/2}dw(t) \\\nonumber
dW(t) =&[rW(t)-c(t)+\Psi(t)SZ(t)]\,dt+\Psi(t)T^{1/2}dw(t) \\ \nonumber
0 = & \lim_{s \to \infty} E\left[J(Z,W,t+s)\right].\nonumber
\end{align}
\begin{theorem}
\label{theorem-optimization}
The investors' conjectures the following form of the value function 
\begin{equation}
\label{Investors-Value-Function}
J(Z, W, t)=-e^{-\beta t-r \varphi W+\Phi(Z) - \lambda},
\end{equation}
where $\Phi(Z)=\frac{1}{2}Z^{\top}LZ$, the optimal share of the stock is
\begin{equation}
\label{optimal-demand}
\tilde{\Psi}(t)=-\frac{T^{1/2}\left(B^{1/2}\right)^{\top}L-S}{r\varphi T}Z(t),
\end{equation}
and the optimal consumption is given by
\begin{equation}
\label{optimal-consumption}
\tilde{c}(t)=\frac{\frac{1}{2}Z^{\top}(t)LZ^{\top}(t)+r\varphi W(t)+\lambda-\ln(r)}{\varphi},
\end{equation}
where $L\equiv(l_{i,j})_{i,j=1}^{5}$ is a symmetric real matrix and $\lambda$ is a real number satisfying
\begin{equation}
\label{Xi-lambda-equ.}
r[1+\lambda-\log(r)]-\beta-\frac{1}{2}\mathbf{tr}\left[\left(B^{1/2}\right)^{\top}LB^{1/2}\right]=0.
\end{equation}
The investors' demand and consumption equation are optimal when coefficients of matrix $L$ are solutions of the following algebraic Riccati equation
\begin{equation}
\label{Xi-L-equ.}
LUL-LX-X^{\top}L-Y=0,
\end{equation}
for
\begin{align}
\label{U-X-Y-equ.}\nonumber
U& \equiv B^{1/2}\left[TI_{4}-\left(T^{1/2}\right)^{\top}T^{1/2}\right]\left(B^{1/2}\right)^{\top}\\ 
X& \equiv T\left(A-\tfrac{1}{2}rI_{5}\right)-B^{1/2}\left(B^{1/2}\right)^{\top}S\\\nonumber
Y& \equiv S^{\top}S
\end{align}
and $I_{n}$ is the identity matrix with dimensions $n$.
\end{theorem}

\begin{proof} See Appendix A.
\end{proof}

Note the risk-less asset is assumed to be constant because investors search stationary equilibrium price when variables are not time-varying. In turn, the interest rate directly affects the price coefficient to account variation in the macroeconomic environment\footnote{\, A more intriguing issue arise when assuming a risk-less asset that depends on the nature of the equilibrium, so the interest rate depends on the type of efficient or inefficient equilibrium price. The work avoids to treat that issue due to the ambitious work of determining the solution of the Riccati equation when some parameters depends on the nature of the equilibrium.}.

\paragraph*{Market clearing.} Market clearing is the condition ensuring the conjectured price in Equation (\ref{Price-Conjectured}) is the equilibrium price. The condition constraints the investors' demand to equalize the stochastic risky asset supply. Assuming the number of informed and noise traders grow at rate $\xi$ and normalizing the initial population of each group to unity, the per capita amount of share supplied is equal to the share per informed traders (see C.K. (1993)). Therefore investors' demand must sum to $1+\Theta$ when market clearing condition applies 
\begin{equation} 
\label{market-clearing}
\tilde{\Psi}(t)=1+\Theta(t). 
\end{equation}
According to Equation (\ref{market-clearing}) the coefficients of $\tilde{\Psi}(t)$ must satisfy the following equalities
\begin{equation} 
\label{investor-strategy}
\psi_{0}=1,\quad \psi_{D_{0}}=0,\quad \psi_{D_{1}}=0,\quad \psi_{D_{I}}=0, \quad \psi_{\Theta}=1
\end{equation}
that are used in the optimization problem (\ref{optimal-demand}--\ref{U-X-Y-equ.}) to determine the price coefficients $[p_{0}, p_{D_0},$ $ p_{D_1}, p_{I}]$. Equation (\ref{U-X-Y-equ.}) verifies that the conjectured form of the equilibrium price (\ref{Price-Conjectured}) is the optimal one. On this account the investors' optimization problem should be more appropriately interpreted as the determination of the risky asset price that makes the rational investors' equilibrium demand for the risky asset optimal. Theorem (\ref{theorem-optimization}) and market clearing condition (\ref{market-clearing}) determine an equilibrium price of two types: the efficient price (Equilibrium-Type A) and inefficient price (Equilibrium-Type B). The latter differs from the former when at least one of its coefficient deviates from the ones determined by the efficient condition.

\paragraph*{Equilibrium-Type A: Efficient Equilibrium Price}

The efficient equilibrium price states the fundamental value $V(t)$ is the expected present value of dividend and non-dividend at the risk-less rate $r$.
\begin{proposition}
The economy defined in Equations (\ref{Dividend-equ.}--\ref{cara-utility}) has a stationary rational expectations equilibrium in which price is efficient:
\begin{align}
\label{Price-Equilibrium-Type-A}
\tilde{P}(t)=\:& \tilde{V}(t) + \tilde{p}_{0}+\Theta(t)\\ \nonumber
        =\:&\tilde{p}_{0}+\tilde{p}_{D_0}D_{0}(t)+\tilde{p}_{D_{1}}D_{1}(t)+\tilde{p}_{I}I(t)+\Theta(t)
\end{align}
where
\[
\tilde{V}(t) = E_{t} \intop_{s=0}^{\infty} e^{-r s}D^{u}(t+s)\,ds = E_{t} \intop_{s=0}^{\infty}e^{-(r-\xi)s}D(t+s)\,ds 
\]
and price $\tilde{P}(t)$ has the following coefficients:
\begin{align}
\label{V-efficient}
\tilde{p}_{0}=-\left(\frac{[ (r-\xi+\alpha_{I})^{2}- 2 (r-\xi) \alpha_{I}\rho_{I}]\sigma_{0}^{2}} {(r-\xi)^{2}(r-\xi + \alpha_{I})^{2}}+\frac{\sigma_{D}^{2}}{(r-\xi+\alpha_{D})^{2} } \right) \frac{r}{r-\xi}\varphi \\
\tilde{p}_{D_0} \equiv  \frac{1}{r-\xi}, \qquad \tilde{p}_{D_{1}} \equiv \frac{1}{r-\xi+\alpha_{D}}, \qquad \tilde{p}_{I}\equiv \frac{1}{r-\xi}-\frac{1}{r-\xi+\alpha_{I}}
\end{align}
\end{proposition}

\begin{proof} See Appendix B.
\end{proof}

The constant term is obtained when informed investors maximize their objective function (\ref{Investors'-objective-function}) using the price coefficients in the form of Equation (\ref{V-efficient}). Note that $\tilde{p}_{0}$ depends on the fundamental parameters $\left(\alpha_{D},\,\alpha_{I},\,\rho_{I},\,\sigma_{0}, \,\sigma_{D}\right)$ and mainly by the investors' risk aversion $\varphi$ which affects the expected return on the stock by increasing risk premium with a higher negative term. As shown in C.K. (1993) and Wang (1993), this is a simple discount on the price to account for the increasing discount rate. Finally note that it is possible to extract from (\ref{V-efficient}) the algebraic form of $\gamma$ given in Equation (\ref{Price-Equation}) and, indirectly, having a measure of the investors' risk aversion using real data.

\paragraph*{Equilibrium-Type B: Inefficient Equilibrium Prices}

An informationally inefficient equilibrium price is any equilibrium price who deviates from its fundamental value, defined in Equilibrium-Type A. The economy defined in Equations (\ref{Dividend-equ.}--\ref{cara-utility}) has a stationary REE in which the equilibrium price is informationally inefficient given by
\begin{equation}
\widehat{P}=\widehat{p}_{0}+\widehat{p}_{D_{0}}D_{0}(t)+\widehat{p}_{D_{1}}D_{1}(t)+\widehat{p}_{I}I(t)+\Theta(t)
\end{equation}
when at least one of the price coefficients is not equal the efficient ones shown in Equation (\ref{V-efficient}).

\subsection{Utility Function}
There are several candidates equilibrium price, which correspond to Equilibrium-Type A or Equilibrium-Type B, for any given set of exogenous parameters. Each candidate is solution of the infinite-horizon optimization problem and the Walrasian auctioneer Pareto rank the candidate equilibrium prices according to the Utility criteria. Finally, it considers the one with the highest Utility for the informed investors as the equilibrium price.
The value function (\ref{Investors-Value-Function}) has the following form when Theorem (\ref{theorem-optimization}) is verified:
\begin{align}
J(Z,L,\lambda,t) & =\lambda+\frac{1}{2}l_{1,1}+f[D_{0}(t),D_{1}(t),I(t),\Theta(t)]
\end{align}
where the constant term $\lambda+\frac{1}{2}l_{1,1}$ is different for each candidate equilibrium price and represents the ``essential Utility part'' of the value function. The ``essential Utility part'' selects among different equilibrium candidates is $\tilde{J}(L, \lambda)  \equiv \lambda+\frac{1}{2}l_{1,1}$. Since we are in a suitable neighborhood of the origin of the Euclidean space of the states of the economy $(D_{0},D_{1},I,\Theta)$ we have $Z L Z \simeq l_{1,1}$, such that 
\[
-e^{-\beta t-r \varphi Y+\frac{1}{2}Z^{\top}LZ - \lambda} \simeq \beta t + r \varphi Y + \tilde{J}(L, \lambda).
\]
The higher is the ``essential Utility part'' of the expected Utility, the higher is the expected Utility itself in the considered neighborhood.
\medskip

To put in evidence how result is obtained and what is the corresponding economic interpretation. 
Under a technical point of view, result is driven by a maximization of the investment problem using a price a-priori inefficient (Equilibrium-Type B) and different market conditions. For any market condition there are several candidates equilibrium price due to the high non-linearity of equation (\ref{Xi-L-equ.} - \ref{U-X-Y-equ.}). The ``essential Utility part'' helps to select among them the price giving the investors the highest Utility. As long as risk-aversion and noise volatility increases the Equilibrium-Type B Pareto dominates the Equilibrium-Type A.

The economic interpretation of this result is as follows: (i) investors always find the profitable price for their investment problem, given any market conditions; (ii) as long as they perceive a \textit{low} level of risk in the market, they get more Utility trading with noise an efficient price; (iii) in turn, a change in risk perception lead the profitable condition to an inefficient price; (iv) investor ask a higher risk-premium as a compensation of the higher perceived risk. Last, only a change of risk perception can modify the obtained Equilibrium-Type from the efficient to the inefficient one, and the contrary.

As a consequence of this result, there exists a critical threshold of price variance, below (above) which price converges to an informationally efficient (inefficient) equilibrium. The threshold depends on private and public information, and on the noise asset volatility. According to our results, when noise volatility represents a bigger proportion of price volatility price loose is efficiency and investors ask a bigger risk premium with respect to the one asked if the efficient price had been the equilibrium one. 

\section{Numerical Solutions}
\label{Numerical-Solution}
This section shows with numerical routines the candidates equilibrium price and the correspondent ``essential Utility part.'' 
The exogenous parameters of the model are as follow\footnote{\, The choice of the parameters follows Wang (1993)}:
\[
\begin{array}{c}
r=0.05; \quad \xi=0.011; \quad \beta=0.30; \quad \varphi=0.50; \\
\alpha_{D}=0.50; \quad \alpha_{I}=0.40;  \quad \alpha_{\Theta}=0.05; \\
\sigma_{0}=0.50; \quad \sigma_{D}=0.10; \quad \sigma_{I}=0.40; \quad \sigma_{\Theta}=0.50.
\end{array}
\]
Table I shows several candidates equilibrium price and the one with the highest ``essential Utility part'', the equilibrium price in in grey. It is important to stress that the equilibrium price obtained numerically is the same one derived algebraically in Proposition (1) where
\begin{align*}
\tilde{p}_{0}=&-91.773, \qquad \qquad \qquad \qquad \qquad\, \tilde{p}_{D_{0}}= \frac{1}{r-\xi} = 25.641,\\
\tilde{p}_{D_{1}}=&\frac{1}{r-\xi+\alpha_{D}} = 1.855,\qquad \qquad \qquad  \tilde{p}_{I}=\frac{\alpha_{I}}{(r-\xi)(r-\xi+\alpha_{I})}=18.446.
\end{align*}

\medskip
\begin{center}
[Insert Table I]
\end{center}
\medskip

The efficient equilibrium price Pareto dominates the inefficient ones. The constant term $\tilde{p}_{0}=-91.773$ is negative, as requested by the model to account for the expected return on the stock for risk-averse informed investors. The other candidates equilibrium price are inefficient because they present a lower (higher) value of $p_{I}$ or $p_{D_{0}}$ with respect to the full-informative price. As a consequence of this inefficiency investors request a higher (lower) value of discount term (risk premium) as a compensation for inefficiency. 

Table II shows the candidate equilibrium prices when investors' risk aversion changes to unity.\footnote{\, We double the risk aversion parameter with respect to Table I} Among the candidate equilibrium prices, the one with the highest ``essential Utility part'' ($\lambda + \frac{1}{2}l_{1,1}=43.96$) is an inefficient one with coefficients:
\[
\widehat{p}_{0}= -2664.632, \quad \widehat{p}_{D_{0}}= -89.311,\quad \widehat{p}_{D_{1}}=1.855,\quad \widehat{p_{I}}=-13.384
\]
in which $\widehat{p}_{D_{0}}=-89.311$ and $\widehat{p_{I}}=-13.384$, the coefficients associated with the permanent component of dividend and with private information, deviate from their efficient values given by $\tilde{p}_{D_{0}}= \frac{1}{r-\xi} = 25.641$ and $\tilde{p}_{I}=\frac{\alpha_{I}}{(r-\xi)(r-\xi+\alpha_{I})}=18.446$. 

\medskip
\begin{center}
[Insert Table II]
\medskip
\end{center}

The negative coefficients represents the investors' willingness to negatively correlate the private and public information flows with price. Finally, investors ask a higher discount: $\widehat{p}_{0}= -2,664.632$ as a compensation for the inefficient equilibrium price achieved: the higher discounted term demanded might explain the well-known equity risk premium.

Table III shows an analogous result when noise volatility increases, i.e., $\sigma_{\Theta}$ is unity. The equilibrium price is characterized by the following coefficients:
\[
\widehat{p}_{0}= -465.202, \quad \widehat{p}_{D_{0}}= -80.445,\quad \widehat{p}_{D_{1}}=1.855,\quad \widehat{p_{I}}=-87.639
\]
where the permanent component coefficient $\widehat{p}_{D_{0}}=-80.445$ and private information coefficient $\widehat{p_{I}}=-87.639$ deviate from the efficient reference values. 

\medskip
\begin{center}
[Insert Table III]
\medskip
\end{center}

There is still a negative value of $p_{D_{0}}$ to account that investors' demand is anti-correlated with the permanent component of dividend. To remark that it does not imply that price moves downward when dividend increases because the coefficient measures the price reaction to dividend news. On the other hand, the price dynamics is mainly driven by the constant part. 

Figure 1 (\ref{Figure1}) shows the ``essential Utility part'' of the efficient (Equilibrium-Type A) and inefficient (Equilibrium-Type B) prices as function
 of risk aversion and noise volatility. For $\varphi =0.1$ and $\sigma_{\Theta}=0.1$ the equilibrium price is the efficient one as shown by the higher Utility (red area). An increase of noise volatility to $\varphi = 3.0$ let the efficient one still Pareto dominate the inefficient ones. Contrariwise, with $\varphi = 0.5$ the equilibrium price depends only on the values of noise volatility. A further conclusion is that risk aversion has a stronger effect on that result with respect to noise volatility.

\paragraph*{Robustness check.} Numerical results are obtained letting the machine search for mathematical solutions using the Newtonian method. We initialize the research letting the price coefficients and each element of Equation (\ref{U-X-Y-equ.}) range over $(-10;10)$. All equations and starting values are real, and the research is given only for real roots.  Results are controlled by expanding the range and using the method of secant. A significant change in the exogenous parameters of the model does not impact the final result. Finally, a combination of high risk aversion and high price volatility solely lead the result. This is verified increasing the dividend ($\sigma_{0}, \sigma_{D}$) and private information volatility ($\sigma_{I}$).  

\section{Econometric Method and Empirical Results}
\label{estimation}
\subsection{Preliminary Data Analysis}

The estimates of the model are performed using the annual U.S. time-series data taken from Shiller (2000) and used in other works such as C.K. (1993) and Campbell and Shiller (1987, 1988). Dataset consists of monthly stock price, the corresponding dividend data, and a price index during the period January 1871--December 2009.\footnote{\,  The dataset was retrieved from Robert Shiller's website at http://www.econ.yale.edu/shiller/. } The real stock price is the Standard and Poors Composite Stock Price Index multiplied by the CPI-U (Consumer Price Index-All Urban Consumers) in June 2010, and divided by the corresponding CPI. The procedure is applied to the corresponding S\&P Stock Price dividend per share to obtain the real dividend series (see, Shiller (2000)). January of each year is considered as the annual data point to avoid problems with time aggregation.

The real unadjusted price and dividend are $P_{t}^{u}$ and $D_{t}^{u}$ to distinguish from the de-trended values $P_{t}$ and $D_{t}$, as in Equation (\ref{D_U-P_U}). The value $\xi=0.0115$ is the average mean dividend growth rate over the all sample. The de-trended operation aims to remove the exponential growth from the \textit{ex-ante} mean of data without forcing data to revert to a trend line. Secondly, it gets rid of the exponential growth from the variance of data, a rescaling effect similar to a logarithm transformation. Finally, $D_{t}$ and $P_{t}$ are normalized such that the mean of price is equal to one by dividing each time series with the mean of price. Figure 1 plots the de-trended real price and dividend $\times$ 10.

Table IV presents the main statistics of dataset and Table V presents the results of ADF, PP, and KPPS test for stationarity both for $P_{t}$ and $D_{t}$, and for $P_{t}^{u}$ and $D_{t}^{u}$. 

\medskip
\begin{center}
[Insert Table IV]

[Insert Table V]
\medskip
\end{center}

The same test are computed for ln$(P_{t}^{u})$ and ln$(D_{t}^{u})$ as a control. Unit roots test show that price and dividend have unit roots in the first level. The results for ADF, with five lags, and PP test for $P_{t}$, $P_{t}^{u}$, and ln$(P_{t}^{u})$ do not always reject the null hypothesis of nonstationarity, while KPPS always rejects stationarity. The ADF test for $D_{t}$, $D_{t}^{u}$, and ln$(D_{t}^{u})$ does not reject the null in presence of a trend at the 5\% level. KPSS confirms the unit root in level of dividend. The PP test accounts for autocorrelation of the error term and, it shows that dividend is stationary with trend.\footnote{\, Assuming a unit root in dividend implies that noise does not help to explain the stock price volatility (as shown by Kleidon (1986), Marsh and Merton (1986), and Timmermann (1996)).} In light of the results above, it is assumed a unit root in the stock price and dividend time series, as shown in C.K. (1993) and Campbell and Shiller (1987, 1988). Finally note that the de-trended operation does not have any effect on the unit root assumption. 

Table VI reports other time series properties of data such as the sample correlation of the integrated process $\Delta P_{t}$, $\Delta D_{t}$ until the fifth lags and, in the bottom of the table, the sample standard deviations of $\Delta P_{t}$ and $\Delta D_{t}$. 

\medskip
\begin{center}
[Insert Table VI]
\medskip
\end{center}

The correlation analysis suggests that price and dividend have a mean reverting component: this is due to the positive first autocorrelation (0.14 and 0.22, respectively) while the other autocorrelations are negatives.\footnote{\, The autocorrelation sample in 1871--1988 supports the mean reverting component only for dividend, as C.K. (1993) show in their work.} The result supports the rejection of the null hypothesis of unit root for dividend. The cross correlation between the dividend change from the end of one year to the end of the next one, $\Delta D_{t}$, and the corresponding price change, $\Delta P_{t}$, shows a very low value at the contemporaneous level (0.03), and at the third level (0.06). On the other hand, there is a high correlation (0.44) between $\Delta D_{t}$ and $\Delta P_{t-1}$, supporting that only price changes between $P_{t}$ and $P_{t-2}$ help to explain changes in actual dividend. In conclusion, there is evidence that an hidden variable, the so called ``private information'' $I_{t}$, might help to explain the relationship between price and dividend. 

The model assumes that price and dividend have one unit root, they are integrated processes of order one $I(1)$, and that a linear combination is stationary [they are $I(0)$]. The existence of cointegration is tested with Engle-Granger two-step method (where the null is no cointegration, and the residual is a random walk). Table VII shows the regression result and the ADF test on the residuals. 

\medskip
\begin{center}
[Insert Table VII]
\medskip
\end{center}

The estimate $D_{t}$ = 0.073 + 0.012$P_{t}$ using the heteroskedasticity-robust standard error reports that two-thirds of dividend's mean is explained by the constant term. The ADF test rejects the unit root hypothesis at $5\%$ supporting the existence of a linear combination of $P_{t}$ and $D_{t}$. In turn, the coefficient regressor ($0.012$) is equal to $(r-\xi$), the interest rate less the dividend growth rate ($\xi$ = 0.011), implies a low interest rate equals to 2.4\%; a low level justified by the high constant term. The regression without a constant gives $D_{t}$ = 0.022$P_{t}$, which implies an interest rate of 3.3\%. Reversing the estimate the regression is $P_{t}$ = $-$3.113 + 57.003$D_{t}$, such that $(r-\xi$) = 0.017 and $r = 2.9\%$. The latter value is closer to the mean rate of return on the stock index $(3.4\%)$. In conclusion, the model considers $r=3\%$ and as a control $r=6\%$ and $r$ free to account different interest rate during the long period used in the estimates.

\subsection{Estimates of the Model}

The model is set up in continuous time while dataset is in discrete time. The use of an exact discrete analog, according to the procedure originally introduced by Bergstrom (1966, 1983) and recently discussed in McCrorie (2009), allows the estimates of the model. A different procedure such as in C.K. (1993) estimates the discrete time transformation of the original continuous time model when it is possible to show that the stacked vector of point-sampled and time-averaged transformation of the continuous-time variables follow a discrete-time AR(1) process. In turn, the choice of the exact discrete analog follows from the benefits of the state space approach: price and dividend are the observed variables, or the measurement equation of the state space model, while $Z(t)$ is the vector of state variables representing the transition equation. The state space dynamics is in compact form as
\begin{align}
\label{compact-form-1}
Y_{t} =\:& C(\mu) Z_{t}\\ \nonumber
dZ_{t} =\:& A(\theta)Z_{t}dt + B^{1/2}(\theta)\,dw(t)
\end{align}
where the vector $Y_{t}=[P_{t},D_{t}]$, the matrix $C(\mu)$ contains the price coefficients $[p_{0}, \, p_{D_{0}}, \, p_{D_{1}}, \, p_{I}, \, 1]$, $\{A(\theta)$, $B(\theta)\}$ are matrices containing the unknown parameters $\theta=[\alpha_{D},\, \alpha_{I},\, \alpha_{\Theta}, \, \sigma_{0}, \, \sigma_{D}, \, \sigma_{I}, \, \sigma_{\Theta}]$ and, $dw(t)$ is the vector of independent brownian motions. Appendix B reports the exact discrete matrices of the continuous time model and the corresponding likelihood function. Kalman filter is used to extract the hidden information and to compute the estimates of the parameters. Data are non-stationary such that it is used a non-informative (diffuse) prior distribution for the corresponding parameters. The filter is initialized assuming the unconditional mean of each state variables as zero and an arbitrarily large covariance matrix as suggested by De Jong (1991).\footnote{\, The variance covariance matrix is a diagonal matrix with $10^{6}$ on the diagonal.}

\subparagraph{Data 1871--2009}
Table VIII shows the Maximum Log-Likelihood (ML) values of the estimates of model A, where it is assumed the informationally efficient price form as given in Equation (\ref{V-efficient}), and Equilibrium-Type B, in which there are not constraint into the price coefficients. 

\medskip
\begin{center}
[Insert Table VIII]
\medskip
\end{center}

Each row of the table presents different assumptions about the interest rate, while the columns report the absence or the presence of market noise. The main findings are: (i) Equilibrium-Type A has higher MLs for any assumption about interest rate even assuming there are no noise traders in the market; (ii)  noise traders strongly improve the goodness of the model, as shown in C.K. (1993), when interest rate equals $3\%$ and $6\%$; (iii) the assumption of  full noise and interest rate equals $3\%$ let the estimate of Equilibrium-Type A (ML equals 834.88) be higher than Equilibrium-Type B (ML equals 830.65); (iv) the previous result is confirmed when interest rate is assumed equals to $6\%$ or when we estimate it separately. 

Table IX shows the estimates of price coefficients and two results are straightforward: (i) the constant term $p_{0}$ is negative as shown by the theoretical model to account for the risk aversion of the investors; (ii) the estimate of the interest rate equals $3.09\%$ (with ML = $ 831.62$) when we assume an efficient price and full noise; this is in line with the implied interest rate obtained in the regression.

\medskip
\begin{center}
[Insert Table IX]
\medskip
\end{center}

Finally, the overall estimates reflect that S\&P 500 Index correctly reports the fundamental value during the period 1871--2009 and the EMH prevailed. It is important to emphasize that during the long period there were at least two structural breaks that are not considered in the estimates of the model. In turn, the estimates are used to confirm the informational efficiency of the S\&P 500 Index and, to compare the results with those of C.K. (1993).

\subparagraph{Dot.com Bubble}
The theoretical results given in Section \ref{Numerical-Solution} have some evidences on real financial data using monthly prices and dividends of the S\&P 500 Index during 1995--2000. Table X shows unit root tests values at 5\% level: price and dividend have unit roots in level and they are stationary if differenced of order one. In the bottom of Table X we test the cointegration using Johansen's methodology. 

\medskip
\begin{center}
[Insert Table X]
\medskip
\end{center}

The values of $\lambda \textrm{max} = 30.55 \, (15.67)$ and $\lambda \textrm{trace} = 31.79 \, (19.96)$ reject the null hypothesis at 5\% level of zero cointegrating vector and they accept the hypothesis of one vector of cointegration. The regression price on dividend is $D_{t} = 0.049 + 0.0017 P_{t}$, using the heteroskedasticity-robust standard error, therefore using $\xi= 0.0011$ the implied interest rate is $r=0.0028$.  As C.K. (1993) have shown, the low interest rate value is justified by the high value of $\lambda$. The regression without the constant let $r=1.5\%$ that is more suitable. In any way, the model estimates are performed assuming the interest rate equals to $r=1.5\%, 3\%$, and be free.

Table XI presents the main empirical results: the ML of Equilibrium-Type B is always higher than ML of Equilibrium-Type A. The higher ML is obtained using $r=1.5\%$ that is closer to the implied interest rate. 

\medskip
\begin{center}
[Insert Table XI]
\medskip
\end{center}

The Likelihood Ratio test is used to test the fitting ability of the two models. In turn, LR-test rejects at $1\%$ the null hypothesis of Equilibrium-Type A, in favor of Equilibrium-Type B, for any given assumption of the interest rate. 

Table XII reports the price coefficients of Equilibrium-Type A and Equilibrium-Type B for different assumptions about interest rate. 

\medskip
\begin{center}
[Insert Table XII]
\medskip
\end{center}

There is evidence that the constant term is negative as suggested by the model and, the estimate of $r$ when is free is very low ($r=0.002$). This is in line with the preliminary analysis, although  the assumption $r=1.5\%$ has been considered more reliable. 

The estimates of the parameters shown in Table XIII support the theoretical results: (i) the constant term estimated according to Equilibrium-Type B is lower ($p_{0}=-0.061$) with respect to the constant term of Equilibrium-Type A ($p_{0}=-0.177$) thus risk-averse investors demand a lower risk premium when price is inefficient and, as a consequence, price increases (as the positive asset bubble shown in Figure 2); (ii) the permanent component of the dividend process is informationally efficient ($p_{D_{0}}=71.942$ and $p_{D_{0}}=71.944$ in Equilibrium-Type A and Equilibrium-Type B, respectively) while investors estimate a lower value of their private information ($p_{I}=0.001$) and the permanent component of dividend ($P_{D_{1}}=1.089$): this is confirmed with the calibration exercised computed in Section (\ref{Numerical-Solution}); (iii) the low value of $p_{I}$ shows that private information decays very fast ($\alpha_{I}=0.905$): investors cannot exploit their private information for a long period; (iv) the inefficient price has the highest ML when investors has a high noise volatility with respect to dividend and information: $\sigma_{\Theta}=0.124$ (see Table XIII), an evidence of the theoretical results. 

\medskip
\begin{center}
[Insert Table XIII]
\medskip
\end{center}

\section{Conclusion}
\label{Conclusion}

The main result of this work is the existence of two types of equilibrium prices both solutions of the optimal investment of the informed risk-averse investors: the semi-strong efficient one, in which price reflects the fundamental asset value, and the inefficient ones where the quality informativeness of price is lost. The model shows that the equilibrium price, selected among multiple candidate prices as the one with highest utility for the informed investors, is determined by the investors' market risk perception which is measured in terms of risk aversion and noise volatility. The efficient price Pareto dominates (is dominated by) the inefficient one when risk aversion is \textit{low} (\textit{high}) and when noise volatility is \textit{low} (\textit{high}). In conclusion, a change in the market risk perception drives the type of equilibrium from efficiency to inefficiency.  

The second result of the work is the existence of a critical threshold of price volatility below (above) which investors prefer a price informationally efficient (inefficient). The threshold depends on the informative component owned by the informed investor, the private and public information, and depends on the noise contribution to asset volatility. As long as the fundamental value is the main driver of the asset volatility the EMH holds; contrariwise, an higher proportion of noise driving the price volatility lead investors to prefer an informationally inefficient price.

The third result of this work is that using real data, the estimates of the model confirm the theoretical findings. The S\&P 500 Index reflects the fundamental value, given by dividend and hidden private information, during the long period 1871--2009. On the other hand, data do not support the market efficiency hypothesis in the sub-period 1995--2000 and they endorse the main models' theoretical findings. The higher maximum log-likelihood of the model under inefficient price assumption rejects, using the likelihood ratio test, the null hypothesis of an efficient market at 1\%. The result is confirmed under different assumptions of interest rate. Furthermore (i) the estimated inefficient price shows that investors demand a lower risk premium, which according to the theory leads to a price increase shown in the Dot.com bubble; (ii) investors estimate a very fast decay of the private information on price, which leads to lower sensitivity of their private information on the price supporting the role of noise; (iii) a high value of noise volatility supports the developed theory.

In conclusion, the inefficient equilibrium price captures some market anomalies observed in a real financial market when higher risk is involved: (i) the equity risk premium puzzle is explained by a higher discount rate, requested by the risk averse investors when inefficient price hold in the market; (ii) the higher discount rate moves in excess the price volatility with respect to the one observed whether EMH had held.

\newpage{}
\section*{Appendix A. Solution to Investors's Optimization Problem}
\textbf{Proof of Theorem 1.} The investor's optimization problem is solved conjecturing the value function (\ref{Investors-Value-Function}) is the investor's objective function (\ref{Investors'-objective-function}). The value function has the following form
\[
J(Z, W, t)=-e^{-(\beta t + r \varphi W - \Phi(Z) + \lambda)},
\]
where $\Phi(Z)=\frac{1}{2}Z^{\top}LZ$, $Z=\left(1,\, D_{0},\, D_{1}, \, I,\,\Theta \right)^{\top}$ is defined as the $(5$x$1)$ vector of the state variables and $L \equiv (l_{j,k})_{j,k=1}^{5}$. The dynamics of $Z(t)$ is written in compact form as in equation (\ref{compact-form}). The excess return per one share of stock is written in terms of the state vector $P(t)=\bar{P}Z(t)$ such that 
\begin{align*}
dQ=&(D(t)-(r-\xi)P(t))dt+dP(t) =(D(t)-(r-\xi)\bar{P}Z(t))dt+\bar{P}AZ(t)dt+\bar{P}B^{1/2}dw(t)\\
     =& SZ(t)dt+T^{1/2}dw(t)
\end{align*} 
where 
\[
S(t)\equiv M-\bar{P}(r-\xi)+\bar{P}A,\qquad M\equiv(0,1,1,0,0),\qquad\bar{P}\equiv(p_{0},p_{D_{0}},p_{D_{1}},p_{I},1),\qquad T^{1/2}\equiv\bar{P}B^{1/2}\] 
It is proved that
\begin{description}
\item[(i)] the function (\ref{Investors'-objective-function}) is a solution to the Bellman equation
\begin{equation}
\label{uninf.-Bellman-equ.}
\partial_{t}J(Z,W,t)+\max_{\Psi,c}\{\mathcal{G} J(Z,W,t)-e^{-(\beta t+\varphi c)}\}=0,
\end{equation}
where $\mathcal{G}$ is the infinitesimal generator of the diffusion process $(Z(t),W(t))$;
\item[(ii)] the control $(\mathring{\Psi}(t),\mathring{c}(t))$ satisfies
\[
(\mathring{\Psi}(t),\mathring{c}(t))\in\arg\max\{\mathcal{G} J(\mathring{Z}(t),\mathring{W}(t),t)-e^{-(\beta t+\varphi \mathring{c}(t) )}\}=0,
\]
where $(\mathring{Z}(t),\mathring{W}(t))$ is a solution to 
\begin{align}
\label{dZ-dW}
dZ(t)  =&AZ(t)dt+B^{1/2}dw(t)\\
dW(t) =&(rW(t)-c(t)+\Psi(t)SZ(t))dt+\Psi(t)T^{1/2} dw(t)\nonumber
\end{align}
corresponding to the choice of $(\mathring{\Psi}(t),\mathring{c}(t))$;
\item[(ii)] the trasversality condition
\begin{equation}
\label{trasversality-condition}
\lim_{T\rightarrow+\infty}\mathbf{E}_{Z,W,t} \left[ J(t+T, \mathring{Z} (t+T),\mathring{W}(t+T))\right]  =0,
\end{equation}
where $(\mathring{Z}(t),\mathring{W}(t))$ is a solution to (\ref{dZ-dW}).
\end{description}

To show that $J(Z,W,t)=-e^{-(\beta t + \Phi(Z) +\varphi r W + \lambda)}$ is a solution of (\ref{uninf.-Bellman-equ.}), the operator $\mathcal{G}$ is determined:
\begin{equation}
\begin{array}{c}
\label{uninf.-inf.-gen.}
\mathcal{G}  \equiv \frac{1}{2}\sum_{j,k=1}^{5}B_{j,k} \, \partial_{Z_{j},Z_{k}}^{2}+\Psi\sum_{j=1}^{5}T^{1/2}\left(B^{1/2}\right)_{j}^{\top}\,\partial_{W,Z_{j}}^{2}+\frac{1}{2}\Psi^{2}\,T\, \partial_{W,W}^{2} \\
 +\sum_{j=1}^{5}(AZ)_{j}\,\partial_{Z_{j}}+\left(  rW-c-\Psi S Z\right)
\,\partial_{W}.
\end{array}
\end{equation}
On the other hand, using $J$ as a shorthand for $J(Z, W, t)$
\begin{align*}
\partial_{Z_{j}}J    =-(Z^{\top}L)_{j} J, & \qquad \partial_{W}J =-r\varphi J, \qquad \partial_{Z_{j},Z_{k}}^{2}J =\left( LZZ^{\top}L-L\right)  _{j,k}J,\\
\partial_{Z_{j},W}^{2}J  &  =r\varphi(Z^{\top}L)_{j}\,J, \qquad \partial_{W,W}^{2}J  =r^{2}\varphi^{2}\,J.
\end{align*}
Therefore 
\begin{equation}
\begin{array}{c}
\label{uninf.-inf.-gen.-bis}
\mathcal{G} J =\frac{1}{2} \left( \sum_{j,k=1}^{5}B_{j,k}\, \left(LZZ^{\top}L-L\right) _{j,k}\right) \,J +\frac{1}{2}r^{2}\varphi^{2}T\Psi^{2}\,J +  \\
  r \varphi\left(\sum_{j=1}^{5}\left(T^{1/2}\left(  B^{1/2}\right)^{\top}\right)_{j} (Z^{\top}L)_{j}\right)  \Psi J -\left(  \sum_{j=1}^{5}(AZ)_{j} \, (Z^{\top}L)_{j}\right) J-r\varphi\left( r W - c - S Z \Psi \right) J.
\end{array}
\end{equation}
Thanks to the properties of the trace functional
\begin{equation}
\label{rem-1}
 \sum_{j,k=1}^{5}B_{j,k} \left( LZZ^{\top}L-L\right)  _{j,k}\,J=\mathbf{tr}\left((B^{1/2})^{\top} (LZZ^{\top}L-L)B^{1/2}\right) 
 =Z^{\top}LBLZ-\mathbf{tr}\left( \left(B^{1/2}\right)^{\top}LB^{1/2}\right).
\end{equation}
Moreover,
\begin{equation}
\label{rem-2}
\sum_{j=1}^{5}\left(  T^{1/2}\left(  B^{1/2}\right)  ^{\top}\right)_{j}\!(Z^{\top}L)_{j}=T^{1/2}\left(  B^{1/2}\right)  ^{\top}LZ,
\end{equation}
and
\begin{equation}
\label{rem-3}
\sum_{j=1}^{5}(AZ)_{j} \, (Z^{\top}L)_{j}=Z^{\top}LAZ.
\end{equation}
Combining (\ref{uninf.-inf.-gen.-bis}) with (\ref{rem-1})-(\ref{rem-3}), it follows
\begin{equation}
\begin{array}{c}
\label{uninf.-inf.-gen.-ter}
\mathcal{G}J  =\frac{1}{2}Z^{\top}LBLZ-\frac{1}{2}\mathbf{tr}\left(\left(  B^{1/2}\right)  ^{\top}LB_{u}^{1/2}\right)J +r\varphi T^{1/2}\left(  B^{1/2}\right)  ^{\top}LZ\,\Psi J+\frac{1}{2}r^{2}\varphi^{2}T\,\Psi^{2}J \\
 -Z^{\top}LAZ\,J-r\varphi\left(  rW-c-SZ\Psi\right)  \,J.
\end{array}
\end{equation}
The latter, on account of
\[
\partial_{t}J=-\beta J, \qquad Z^{\top}LAZ=\frac{1}{2}(Z^{\top}A^{\top}LZ+Z^{\top}LAZ)
\]
it is possible to rewrite Equation (\ref{uninf.-Bellman-equ.}) in the form
\begin{align}
&  \left(  -\beta+\frac{1}{2}Z^{\top}LBLZ-\frac{1}{2}(Z^{\top}A^{\top}LZ+Z^{\top}LAZ)-\frac{1}{2}\mathbf{tr}\left(  \left(B^{1/2}\right)  ^{\top}LB_{u}^{1/2}\right)  -\varphi r^{2}W\right)J\label{uninf.-Bellman-equ.-bis}\\
&  +\max_{\Psi,c}\{r\varphi((T^{1/2}\left(  B^{1/2}\right)  ^{\top}L+S)Z\Psi+\frac{1}{2}r\varphi T\Psi^{2})J+\varphi rcJ-e^{-(\beta t+\varphi c)}\}\nonumber\\
&  =0.\nonumber
\end{align}
Setting
\[
J\equiv J(t,Z,W,\Psi)\equiv \left( (T^{1/2}\left(  B^{1/2}\right)  ^{\top}L+S)Z\Psi+\frac{1}{2}r\varphi T\Psi^{2}\right)J,
\]
and
\[
K\equiv K(t,Z,W,c)\equiv r\varphi cJ-e^{-(\beta t+\varphi c)},
\]
it is possible to get
\begin{align*}
&  \max_{\Psi,c}\{r\varphi((T^{1/2}\left(  B^{1/2}\right)  ^{\top}L+S)Z\Psi+\frac{1}{2}r\varphi T\Psi^{2})J+\varphi rcJ-e^{-(\beta t+\varphi
c)}\}\\
&  =r\varphi\max_{\Psi}\{J(t,Z,W,\Psi)\}+\max_{c}\{K(t,Z,W,c)\}
\end{align*}
Maximizing $J$ [resp. $K$] with respect to $\Psi$, the first conditions are
\begin{align}\nonumber
\frac{dJ}{d\Psi}=&((T^{1/2}\left(  B^{1/2}\right)  ^{\top}L+S)Z+r\varphi\Psi T)J =0 \qquad \qquad \frac{dK}{dc}=r\varphi J+\varphi e^{-(\beta t+\varphi c_{u})} =0 \nonumber
\end{align}
that yields
\[
 \tilde{\Psi}=-\frac{(T^{1/2}\left(  B^{1/2}\right)  ^{\top}L+S)}{r\varphi T}Z \qquad \qquad \tilde{c}=\frac{\frac{1}{2}Z^{\top}LZ+r\varphi W+\lambda-\log(r)}{\varphi}
\]
which is the desired optimal demand (\ref{optimal-demand}) and optimal consumption (\ref{optimal-consumption}). 

Moreover, the second order condition
\[
 \frac{d^{2}J}{d\Psi^{2}}=r\varphi TJ \leq0    \qquad \qquad \qquad \qquad \quad \qquad \frac{d^{2}K}{dc^{2}}=-e^{-\varphi c} \leq0 \nonumber
\]
guarantees that $\Psi$ is optimal for $J$ and, similarly, $c$ is optimal for $K$. 
As a consequence,
\begin{align*}
&  \max_{\Psi}\{J(t,Z,W,\Psi)\} =-\frac{1}{2r\varphi T}\left(  Z^{\top}\left(  LB^{1/2}(T^{1/2})^{\top}+S^{\top}\right)  \left(  T^{1/2}\left(  B^{1/2}\right)^{\top}L+S\right)  Z\right)  J
\end{align*}
and
\begin{align*}
\max_{c}\{K(t,Z,W,c)\}  &  =r\left(  \frac{1}{2}Z^{\top}LZ+r\varphi W+\lambda-\log(r)+1\right)  J.
\end{align*}
In light of what shown above, the Bellman equation (\ref{uninf.-Bellman-equ.-bis}) takes the form
\begin{align*}
&  \left(  -\beta+\frac{1}{2}Z^{\top}LBLZ-\frac{1}{2}(Z^{\top} A^{\top}LZ+Z^{\top}LAZ)-\frac{1}{2}\mathbf{tr}\left(\left(B^{1/2}\right)^{\top}LB_{u}^{1/2}\right)  -\varphi r^{2}W\right) J \\
&  -\frac{1}{2T}\left(  Z^{\top}(LB^{1/2}(T^{1/2})^{\top}+S^{\top}\right)  \left(  T^{1/2}\left(  B^{1/2}\right)^{\top}L+S\right)  ZJ\\
&  +r\left(  \frac{1}{2}Z^{\top}LZ+r\varphi W+\lambda-\log(r)+1\right)J=0\\
\end{align*}
that is
\begin{align}
\label{uninf.-Bellman-equ.-ter}
&  \frac{1}{2}Z^{\top}\left(  LBL-\frac{1}{T}\left(  L^{\top} B^{1/2}(T^{1/2})^{\top}+S^{\top}\right)  \left(  T^{1/2}\left(B^{1/2}\right)  ^{\top}L+S\right) -A^{\top}L-LA+rL\right)
ZJ\\
&  +\left(  r\lambda+r(1-\log(r))-\beta-\frac{1}{2}\mathbf{tr}\left( \left(B^{1/2}\right)  ^{\top}LB^{1/2}\right)  \right)  J=0.\nonumber
\end{align}
On the other hand,
\begin{align}
&  LBL-\frac{1}{T}\left(  LB^{1/2}(T^{1/2})^{\top}+S^{\top}\right) \left(  T^{1/2}\left(  B^{1/2}\right)  ^{\top}L+S\right)  -A^{\top}L-LA+rL\label{rem-4}\\
&  =\frac{1}{T}\left(  LB^{1/2}\left( TI_{5}-\left(  T^{1/2}\right)^{\top}T^{1/2}\right)  \left(  B^{1/2}\right)  ^{\top}L\right)\nonumber\\
&  -\frac{1}{T}\left(L\left(B^{1/2}\left( T^{1/2}\right)^{\top}S+T\left(  A-\frac{1}{2}rI_{5}\right)  \right)+\left(  S^{\top}T^{1/2}\left(B^{1/2}\right)^{\top}+T\left(  A^{\top}-\frac{1}{2}rI_{5}\right)  L\right)  +S^{\top}S\right).\nonumber
\end{align}
Therefore, combining (\ref{uninf.-Bellman-equ.-ter}) with (\ref{rem-4}), it follows that $J(t,Z,W)$ is a solution of the Bellman equation
(\ref{uninf.-Bellman-equ.}), provided that the matrix $L$ and the parameter $\lambda$ are chosen to fulfill (\ref{Xi-L-equ.}) and (\ref{Xi-lambda-equ.}), respectively.

The trasversality condition (\ref{trasversality-condition}) holds true when the It\^{o}'s formula is applied to the identity
\begin{align*}
&  J(t+\Delta t,\mathring{Z}(t+\Delta t),\mathring{W}(t+\Delta t))-J(t,\mathring{Z}(t),\mathring{W}(t)) =\int_{t}^{t+\Delta t}dJ(s,\mathring{Z}(s),\mathring{W}(s)) 
\end{align*}
which allows to write
\begin{align*}
\label{J-Ito-formula}
J(t+\Delta t,\mathring{Z}(t+\Delta t),\mathring{W}(t+\Delta t))-J(t,\mathring{Z}(t),\mathring{W}(t))=&\int_{t}^{t+\Delta t}(\partial_{s}J(s,\mathring{Z}(s),\mathring {W}(s))+\mathcal{G}J(s,\mathring{Z}(s),\mathring{W}(s)))\,ds\\
 &+\int_{t}^{t+\Delta t}B_{Z,W}^{1/2}\nabla_{Z,W}J(s,\mathring {Z}(s),\mathring{W}(s))\,d\tilde{w}(s),
\end{align*}
where $B_{Z,W}^{1/2}$ stands for the diffusion matrix of the process $(\mathring{Z}(s),\mathring{W}(s))$ and $\nabla_{Z,W}$ denotes the gradient operator in the state space of $(\mathring{Z}(s),\mathring{W}(s))$. On the other hand, since $J\left(  t,Z,W\right)  $ is a solution of the Bellman equation (\ref{bellman-equ}) and $(\mathring{Z}(s),\mathring{W}(s))$ corresponds to an optimal control we get
\begin{align*}
&  \int_{t}^{t+\Delta t}(\partial_{s}J(s,\mathring{Z}(s),\mathring{W}(s))+\mathcal{G}J(s,\mathring{Z}(s),\mathring{W}(s)))\,ds =\int_{t}^{t+\Delta t}e^{-(\beta s+\varphi \mathring{c}(s))}\,ds.
\end{align*}
On account of the latter, applying the expectation operator on both the sides of (\ref{J-Ito-formula}) 
\begin{align*}
&  \frac{\mathbf{E}_{t,Z,W}\left[  J(t+\Delta t,\mathring{Z}(t+\Delta t),\mathring{W}(t+\Delta t))\right]  -\mathbf{E}_{t,Z,W}\left[  J(t,\mathring
{Z}(t),\mathring{W}(t))\right]  }{\Delta t}\\
&  =\frac{1}{\Delta t}\mathbf{E}_{t,Z,W}\left[  \int_{t}^{t+\Delta t}e^{-(\beta s+\varphi \mathring{c}(s))}\,ds\right],
\end{align*}
and, passing to the limit as $\Delta t\rightarrow0$, it follows
\[
\frac{d\mathbf{E}_{t,Z,W}\left[  J(t,\mathring{Z}(t),\mathring{W}(t))\right]}{dt}=\mathbf{E}_{t,Z,W}\left[  e^{-(\beta t+\varphi \mathring{c}(t))}\right].
\]
On the other hand, by virtue of $c$-first order condition,
\[
e^{-(\beta t+\varphi \mathring{c}(t))}=-\beta J(t,\mathring{Z}(t),\mathring{W}(t)).
\]
Hence, $\mathbf{E}_{t,Z,W}\left[  J(t,\mathring{Z}(t),\mathring{W}(t))\right] $ satisfies the differential equation
\[
\frac{d\mathbf{E}_{t,Z,W}\left[  J(t,\mathring{Z}(t),\mathring{W}(t))\right]}{dt}=-\beta\mathbf{E}_{t,Z,W}\left[  J(t,\mathring{Z}(t),\mathring
{W}(t))\right],
\]
and the desired trasversality condition follows.

\newpage{}

\section*{Appendix B. Proof of Proposition 1}
The efficient market hypothesis states that price is efficient when it is the expected future discounted dividends

\begin{equation}
P(t)= E\left[ \intop_{s=0}^{\infty} e^{-r s}D^{u}(t+s)ds \mid \mathfrak{F}_{t} \right] = E\left[ \intop_{s=0}^{\infty} e^{-(r-\xi)s}D(t+s)ds \mid \mathfrak{F}_{t} \right]
\end{equation}
where the process $D(t)$, the continuos-time dividend yield of a risky asset, is defined as $D(t)=D_{0}(t)+D_{1}(t)$ and
\begin{align*}
dD_{0}(t) =& \alpha_{I}I(t)dt+\sigma_{0}dw_{0}(t), \\
dD_{1}(t) =&-\alpha_{D}D_{1}(t)dt+\sigma_{D}dw_{D}(t),\\
dI(t)=&-\alpha_{I}I(t)dt+\rho_{I}\sigma_{0}dw_{0}(t)+(2\rho_{I}-\rho_{I}^{2})^{1/2}\sigma_{0}dw_{I}(t)
\end{align*}
Rewrite the dividend and informative signals as a trivariate Ornstein-Uhlenbeck process with $Z=(D_{0}, \, D_{1}, \, I)^{\top}$
\begin{equation}
\label{compact-form-2}
dZ(t)=A_{1}Z(t)dt+B_{1}^{1/2}dw(t)
\end{equation}
where
\begin{align}
\label{A1-B1-Equation}
A_{1} & \equiv \left(\begin{array}{ccc}
0 &  0& \alpha_{I} \\
0 &  -\alpha_{D} & 0\\
 0 & 0 & -\alpha_{I} \\
\end{array}\right),
\quad 
B_{1}^{1/2}\equiv
\left(\begin{array}{ccc}
\sigma_{0} & 0 & 0 \\
0 & \sigma_{D}& 0 \\
-\rho_{I}\sigma_{0} & 0& (2\rho_{I}-\rho_{I}^{2})^{1/2}\sigma_{0} \\
\end{array}\right), 
\quad
dw(t) \equiv \left(\begin{array}{c} dw_{0}\\ dw_{D} \\ dw_{I}
\end{array}\right).
\end{align}
Now $Z(t)$ has the following integral form
\[
Z(t+s) = H(s)Z(t)+\intop_{\tau=0}^{s} e^{A(s-\tau)}B^{1/2}dw(t+\tau),
\]
where $H(s)=e^{As}$. Solving differential equation $d H(s)/ds=AH(s)$, with boundary condition $
\mathcal{B}(0)=\left(\begin{array}{ccc}
 1 & 0 & 0 \\
 0 & 1 & 0 \\ 
 0 & 0 & 1\\
\end{array}\right)
$
yields
\[
H(s)=\left(\begin{array}{ccc}
 1 & 0 & 1-e^{-\alpha_{I}s} \\
 0 & e^{\alpha_{D}s} & 0 \\ 
 0 & 0 & e^{-\alpha_{I}s}\\
 \end{array}\right).
\]
Since $E \left[ D(t)\mid \mathfrak{F}_{t}\right]=D(t)$ we have 
\[
E\left[ D(t+s) \mid \mathfrak{F}_{t} \right] = D_{0}+e^{\alpha_{D}s} D_{1}+(1-e^{-\alpha_{I}s})E\left[ I(t) \mid \mathfrak{F}_{t} \right],
\]
thus 
\begin{align*}
& E\left[ \intop_{s=0}^{\infty} e^{-(r-\xi)s}D(t+s)ds \mid \mathfrak{F}_{t} \right] =  E_{t} \intop_{s=0}^{\infty} e^{-(r-\xi)s}\left \{ D_{0}+e^{\alpha_{D}s}D_{1}+(1-e^{-\alpha_{I}s})E\left[ I(t) \mid \mathfrak{F}_{t} \right] \right \}ds\\
=& \frac{1}{r-\xi} D_{0}(t)+\frac{1}{r-\xi+\alpha_{D}}D_{1}(t)+ \frac{1}{r-\xi}-\frac{1}{r-\xi+\alpha_{I}}\hat{I}(t),
\end{align*}
and the coefficients are
\[
\tilde{p}_{D_0} \equiv  \frac{1}{r-\xi}, \qquad \tilde{p}_{D_{1}} \equiv \frac{1}{r-\xi+\alpha_{D}}, \qquad \tilde{p}_{I}\equiv \frac{1}{r-\xi}-\frac{1}{r-\xi+\alpha_{I}}
\]
The constant term is obtained from Equation (\ref{optimal-demand}) imposing the efficient price described above when market clearing condition (\ref{market-clearing}) holds true
\[
\tilde{\Psi}=-\frac{T^{1/2}\left(B^{1/2}\right)^{\top}L-S}{r\varphi T} = 1
\]
where 
\[
T^{1/2} =\frac{(r-\xi-\alpha_{I}(\rho-1))\sigma_{0}}{(r-\xi)(r-\xi+\alpha_{I})};\qquad B^{1/2}= 0 ; \qquad S = -(r-\xi)p_{0}\\
\]
\[
 T = \left(-\frac{\alpha_{I}^{2} \left(-2+\rho _{I}\right) \rho_{I} \sigma_{0}^{2} }{(r-\xi)^{2} \left(r-\xi +\alpha_{I}\right)^{2}}+\left(\frac{1}{r-\xi} - \frac{\alpha_{I} \rho_{I}}{(r-\xi) \left(r-\xi +\alpha_{I}\right)}  \right)^{2} \sigma_{0}^{2}+\frac{\sigma_{D}^{2}}{\left(r-\xi +\alpha_{D}\right)^{2}}+\sigma_{\Theta}^{2}\right)
\]
that implies
\[
\tilde{p}_{0}=-\left(\frac{( (r-\xi+\alpha_{I})^{2}- 2 (r-\xi) \alpha_{I}\rho_{I})\sigma_{0}^{2}} {(r-\xi)^{2}(r-\xi + \alpha_{I})^{2}}+\frac{\sigma_{D}^{2}}{(r-\xi+\alpha_{D})^{2} } \right) \frac{r}{r-\xi}\varphi,
\]
solution of the proof.

\newpage{}

\section*{Appendix C. From Continuous Time Model to Discrete Time Data} 
The model is in continuous-time while data are discrete. We use Bergstrom's (1984, Thm. 3, p. 1167) solution to reformulate the model such that the discrete version, called the exact discrete solution, satisfies the discrete-time data. As underlined by McCrorie (2009), the exact discrete model differs from conventional discrete-time VAR models in that the coefficient matrix and the covariance matrix are functions of the exponential of a matrix. The exact discrete form is obtained from the solution of the continuous-time model (\ref{compact-form}) given by
\begin{equation}
\label{Exact-Discrete-Model-1}
Z(t) = F(\theta)Z(t-1) + \epsilon_{t}	\qquad	(t=1, \dots, T),
\end{equation}
where
\begin{equation}
\label{Exact-Discrete-Model-2}
F(\theta) = e^{A(\theta)},	\qquad \qquad (t=1,\dots, T),
\end{equation}
and the variance-covariance of the independent Gaussian white noise is the solution of the following integral
\begin{equation}
E(\epsilon_{t} \, \epsilon_{t}^{\top}) = \int_0^{1} e^{rA(\theta)} \Sigma(\mu) \left ( e^{rA(\theta)}\right)^{\top}dr = \Omega(\theta, \mu)
\end{equation}
where we assume constant riskless rate and zero mean and covariance matrices: $E(\epsilon_{t}) = 0,$ and $E(\epsilon_{t}\epsilon_{s}^{\top}) = 0$ with $(s \ne t)$.  

The exact discrete form (\ref{Exact-Discrete-Model-1}-\ref{Exact-Discrete-Model-2}) gives the following matrices 
\begin{equation}
\begin{array}{c}
F(\theta) =
\left(\begin{array}{cccc}
1 & 1 - e^{\alpha_{I}} & 0 & 0\\
0 & e^{\alpha_{I}} & 0 & 0 \\
0 & 0 & e^{\alpha_{D}} & 0 \\
0 & 0 & 0 & e^{\alpha_{\Theta}}\\
\end{array}\right) \\
 \Omega(\theta, \mu) = \left(\begin{array}{cccc}
- \frac{\sigma_{0}^{2}(e^{-2\alpha_{I}} (-1 + e^{\alpha_{I}})^{2}\rho-\alpha_{I})}{\alpha_{I}} & - \frac{\sigma_{0}^{2}e^{-2\alpha_{I}} (-1 + e^{\alpha_{I}})^{2}\rho_{I} }{\alpha_{I}} & 0 & 0\\
- \frac{\sigma_{0}^{2}e^{-2\alpha_{I}} (-1 + e^{\alpha_{I}})^{2}\rho_{I}}{\alpha_{I}} & - \frac{\sigma_{0}^{2}(-1 + e^{-2\alpha_{I}})\rho_{I}}{\alpha_{I}} & 0 & 0\\
0 & 0 & -\frac{(-1+e^{-2\alpha_{D}})\sigma_{D}^{2}}{2 \alpha_{D}} & 0 \\
0 & 0 & 0 & -\frac{(-1+e^{-2\alpha_{\Theta}})\sigma_{\Theta}^{2}}{2 \alpha_{\Theta}}\\
\end{array}\right)
\end{array}
\end{equation}
that is used to estimate the parameters $\theta$ = [$\alpha_{I}$, $\alpha_{D}$, $\alpha_{\Theta}$, $\sigma_{0}$, $\sigma_{D}$, $\sigma_{\Theta}, \rho_{I} $] in the likelihood function.
The estimates of the unobservable state vector $Z(t)$ are based on the information available to time $\mathfrak{F}(t)$, where $\mathfrak{F}(t)$ contains the observed data $y$ until $y(t)$. Kalman filter recursive procedure gets the estimates of the parameters $\theta$ in the state vector. The method of diffuse prior to initialize the Kalman filter is used assuming zero mean for the all state variables and, for the variance-covariance an identity matrix with $10^{6}$ on the diagonal.

The likelihood function of $y(t)$ data has joint density $L = \prod_{t=1}^{T}p(y_{t} \mid \mathfrak{I}_{t-1})$ where $p(y_{t} \mid \mathfrak{F}_{t-1})=N(\widehat{y}_{t \mid t-1},f_{t \mid t-1})$ and $f_{t \mid t-1}=E(y_{t}-\widehat{y}_{t \mid t-1} )(y_{t}-\widehat{y}_{t \mid t-1} )^{\top}$. The log-likelihood is given by 
\begin{equation}
\label{loglikelihood}
ln \, L = -\frac{1}{2} \sum_{t=1}^{T} ln \arrowvert f_{t \mid t-1} \arrowvert -\frac{1}{2} \sum_{t=1}^{T} (y_{t}-\widehat{y}_{t \mid t-1} )^{\top} f_{t \mid t-1}^{-1} (y_{t}-\widehat{y}_{t \mid t-1} ).
\end{equation}
such that the estimates of the parameters $\theta$ is given maximizing the equation (\ref{loglikelihood}).

\newpage{}
 
\section*{Figures}
\label{Figure1}
Figure 1. Market Risk Perception measured by risk aversion ($\varphi$) and noise volatility ($\sigma_{\Theta}$) with the corresponding utility of equilibrium price of Equilibrium-Type A and Equilibrium-Type B.
\bigskip
\begin{center}
\includegraphics[scale=1.4]{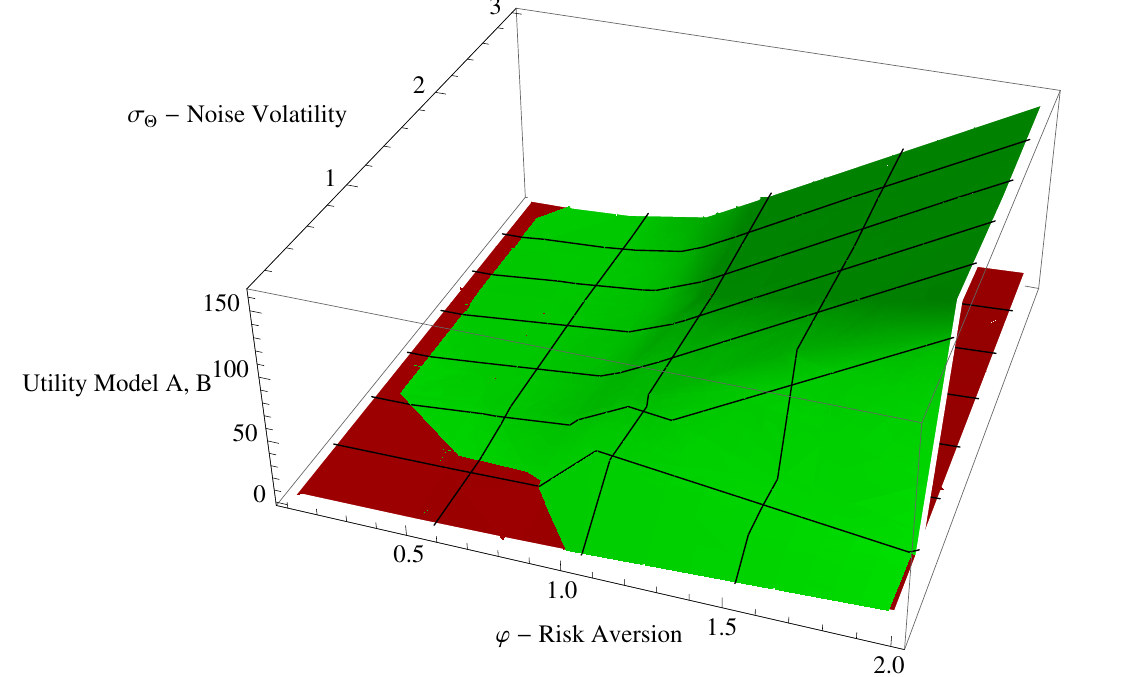}
\end{center}
\footnotesize \textbf{Note} The red area describes the utility values obtained by the informed investors when the efficient equilibrium price holds in the market (Equilibrium-Type A) while green area describes the utility obtained by the informed investors when the inefficient equilibrium price holds in the market (Equilibrium-Type B). The green area is higher than the red one when noise volatility  and risk aversion contemporaneously increase.

\newpage{}
\normalsize
Figure 2. S\&P 500 1850-2005: price and dividend (x10).

\begin{center}
  \includegraphics[scale=0.9]{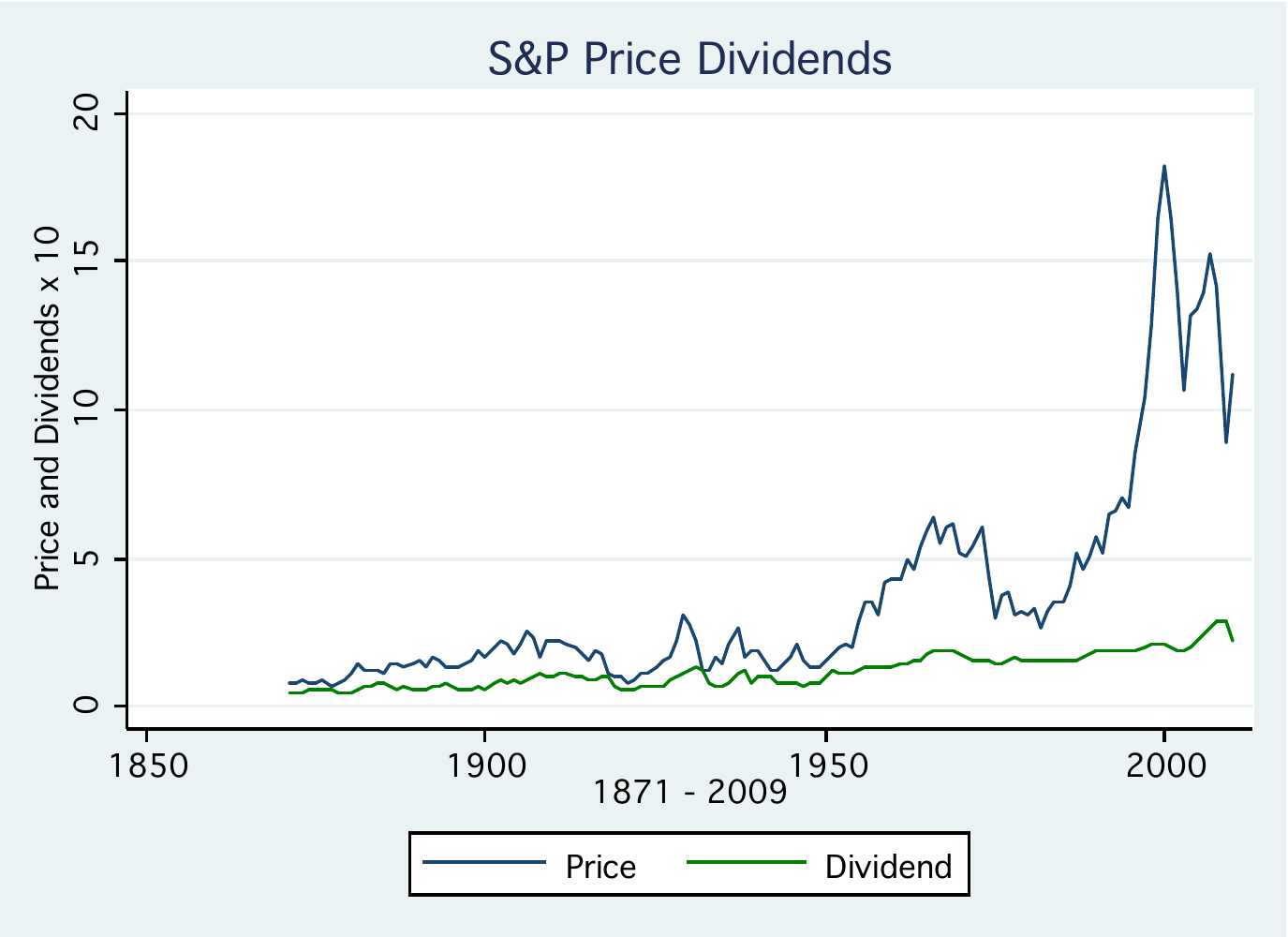}
\end{center}
\footnotesize \textbf{Note} More details on the dataset are available in Shiller's book (2000). According to Shiller (2000) monthly dividend data are computed from the S\&P four-quarter tools for the quarters since 1926, with linear interpolation to monthly figures. Dividend and earnings data before 1926 are from Cowles and associates, interpolated from annual data. Stock price data are monthly averages of daily closing prices through January 2010. The CPI-U (Consumer Price Index-All Urban Consumers) published by the U.S. Bureau of Labor Statistics begins in 1913; for years before 1913 it is spliced to the CPI Warren and Pearson's price index, by multiplying it by the ratio of the indexes in January 1913. December 1999 and January 2000 values for the CPI-U are extrapolated.

\newpage{}

\section*{Tables }
\normalsize
Parameter Set: $r=0.05;$ ~ $\beta=0.30$;~ $\alpha_{D}=0.50;~\alpha_{I}=0.10;~\alpha_{\Theta}=0.05;$ \\
$\quad\,~\sigma_{0}=0.50;~\sigma_{D}=0.10;~\sigma_{I}=1.00; ~\xi=0.011.$\\

\center Table I. Low risk aversion and noise volatility: $~\sigma_{\Theta}=0.50;~\varphi=0.50$. 
\begin{table*}[h!]
\centering
\begin{tabular}{c c c c c c}
\multicolumn{6}{c}{Candidate Equilibrium Price} \\
\hline 
Equilibrium& Utility  & $p_{0}$  & $p_{D_{0}}$  & $p_{D_{1}} $ &$ p_{I} $\\
\hline 
\hline 
\rowcolor[gray]{0.9}   Equilibrium-Type A  & 2.90  & -91.773 & 25.641     & 1.855  & 18.446\\
				Equilibrium-Type B  & 0.97  & -94.218  & 25.641    & 1.855  & -30.080   \\
				Equilibrium-Type B  & 0.95  & -93.148  & 25.641    & 1.855  & -30.897\\
				Equilibrium-Type B  & 0.65  & -82.902  & 25.641    & 1.855  & 51.713\\
				Equilibrium-Type B  & 0.62  & -81.688  & 25.641    & 1.855  & 52.479\\
				Equilibrium-Type B  & 0.41  & -93.635  & 25.641    & 1.855  & 5.775  \\
\hline
\end{tabular}
\captionsetup{margin={1.5cm,1.5cm}, font=footnotesize, labelformat=empty, format=plain, justification=justified}
\caption{Utility value with the corresponding optimal price coefficient for Equilibrium-Type A and Equilibrium-Type B with ``low'' level of risk aversion and noise volatility.}
\end{table*}

\newpage{}

\center Table II. High risk aversion and low noise volatility: $~\sigma_{\Theta}=0.50;~\varphi=1.00$.
\begin{table*}[h!]
\centering
\begin{tabular}{c c c c c c}
\multicolumn{6}{c}{Candidate Equilibrium Price} \\
\hline 
Equilibrium  & Utility  & $p_{0}$  & $p_{D_{0}}$  & $p_{D_{1}} $ &$ p_{I} $\\
\hline 
\hline 
		   		Equilibrium-Type B  & 43.96	 &   -2664.632 &	-89.311      & 1.855    & -13.384   \\
 \rowcolor[gray]{0.9} Equilibrium-Type A  & 4.08	         & 	-183.546 & 	25.641 	 & 1.855	 & 18.446    \\ 
				Equilibrium-Type B  &  3.87	 & 	-188.435 & 	25.641	 & 1.855	 & -30.080   \\
				Equilibrium-Type B  & 3.79	 & 	-186.296 & 	25.641	 & 1.855	 & -30.897   \\
				Equilibrium-Type B  & 3.16	 & 	-187.271 & 	25.641	 & 1.855	 & 5.775      \\
				Equilibrium-Type B  & 2.60	 & 	-165.804 & 	25.641	 & 1.855	 & 51.713    \\
				Equilibrium-Type B  & 2.50	 & 	-163.376 & 	25.641	 & 1.855	 & 52.479    \\
\hline
\end{tabular}
\captionsetup{margin={1.5cm,1.5cm}, font=footnotesize, labelformat=empty, format=plain, justification=justified}
\caption{Utility value with the corresponding optimal price coefficient for Equilibrium-Type A and Equilibrium-Type B with ``high'' level of risk aversion and ``low'' level of noise volatility.}
\end{table*}

\newpage{}

\center Table III. Low risk aversion and high noise volatility: ~$\varphi=0.50$; ~$\sigma_{\Theta}$=1.00.
\begin{table*}[h!]
\centering
\begin{tabular}{c c c c c c}
\multicolumn{6}{c}{Candidate Equilibrium Price} \\
\hline
Equilibrium  & Utility  & $p_{0}$  & $p_{D_{0}}$  & $p_{D_{1}} $ &$ p_{I} $\\
\hline 
\hline
				Equilibrium-Type B  & 8.65  & -465.202  & -80.445    & 1.855 & -87.639\\
				Equilibrium-Type B  & 8.39  & -675.844  & -89.311    & 1.855 & -13.384\\
\rowcolor[gray]{0.9} Equilibrium-Type A  & 2.90  & -91.773  & 25.641  & 1.855 & 18.446 \\
				Equilibrium-Type B  & 1.39  & -91.633  & 25.641    & 1.855 & 18.446\\
				Equilibrium-Type B  & 0.98  & -94.733  & 25.641   & 1.855 & -29.615  \\
				Equilibrium-Type B  & 0.94  & -92.594  & 25.641   & 1.855 & -31.252 \\
				Equilibrium-Type B  & 0.67  & -83.506  & 25.641    & 1.855 & 51.274 \\
				Equilibrium-Type B  & 0.42  & -93.688  & 25.641   & 1.855 & 5.775  \\
\hline 
\end{tabular}
\captionsetup{margin={1.5cm,1.5cm}, font=footnotesize, labelformat=empty, format=plain, justification=justified}
\caption{Utility value with the corresponding optimal price coefficient for Equilibrium-Type A and Equilibrium-Type B with ``low'' level of risk aversion and ``high'' level of noise volatility.}
\end{table*}

\newpage{}
 
\center Table IV. Descriptive Statistics of data.
\begin{table*}[h!]
\centering
\begin{tabular}{c c c c c}
 \hline \hline
Variable  & Mean  & Std. Dev.  & Skewness & Kurtosis \\ 
 \hline
 \\
  $P_{t}$  & 0.994 & 1.002 &2.067 &  6.745   \\ 
  \\
 $D_{t}$  & 0.032  & 0.014 &  0.730 &   2.837 \\ 
 \\
 $P_{t}^{u}$ &  376.379 &  379.342   & 2.067 & 6.745 \\ 
 \\
 $D_{t}^{u}$ &12.127  & 5.642  & 0.725 &  2.83\\ 
 \\
 $ln(P_{t}^{u})$  & -0.362 &  0.798   & 0.650 & 2.607   \\ 
 \\
 $ln(D_{t}^{u})$ & -3.548 & 0.466 & 0.047 & 1.912 \\ 
\hline \hline
\end{tabular}
\captionsetup{margin={1.5cm,1.5cm}, font=footnotesize, labelformat=empty, format=plain, justification=justified}
\caption{$P_{t}$ is the adjusted de-trended price, $P_{t}^{u}$ is real price, $ln(P_{t}^{u})$ is natural logarithm.}
\end{table*}

\newpage{}

 \center Table V. Unit Root Test.
\begin{table*}[h!]
\centering

\begin{tabular}{c c c c c c c}

\hline\hline
	 & \multicolumn{2}{c}{ADF} & \multicolumn{2}{c}{PP} & \multicolumn{2}{c}{KPSS}  	\\
	 		& drift		& trend	&	drift	&trend	&drift		&trend  \\
\hline
	 \\
$P_{t}$ 		&	0.29  	&   -1.11	&  -0.99 	&-2.29	&	1.15* 	&	0.26 	\\
\\
$D_{t}$  		& 	 -1.06 	& -3.46*	&    -1.04	& -3.34* 	&	1.65**	& 0.24	\\
\\
$P_{t}^{u}$	&	0.29  	&   -1.11  	&  -1.00	& -2.28	&	1.15** 	&	0.26\\
\\
$D_{t}^{u}$ 	&	 -1.06 	& -3.46*	&  -1.12	&  -3.34	&	1.65*		& 0.24	\\
\\
$ln(P_{t}^{u})$	&	-0.82  	& -2.29 	& -1.08	-&  -2.63	&	1.48*		&	0.20*\\
\\
$ln(D_{t}^{u})$ 	&	-2.04 	& -4.81** 	& -1.74	&  -4.07**	&	1.69**	&	0.09\\
\\
\hline
\multicolumn{6}{l}{Note: * rejects the at 5 \%, ** rejects at 1\%} \\ 
\hline\hline
\end{tabular}
\captionsetup{margin={1.5cm,1.5cm}, font=footnotesize, labelformat=empty, format=plain, justification=justified}
\caption{The Augmented Dickey-Fuller test for a unit root against the alternative of an explosive root is applied to each time series of real price and real dividend, the de-trended price and de-trended dividend, and the log stock price and log stock dividend. Similar results are obtained using the Phillips-Perron test used to control for unspecified autocorrelation and heteroschedasticity and Kwiatkowski-Phillips-Schmidt-Shin (KPSS) used for testing the stationarity as null hypothesis.}
\end{table*}

\newpage{}

\center Table VI. Time Series Properties of the data.
\begin{table*}[h!]
\centering
\begin{tabular}{c c c c c c c}
 \hline\hline
 Correlations: \quad  & $\Delta D_{t}$  & $\Delta P_{t}$ & \quad  &  & $\Delta D_{t}$  & $\Delta P_{t}$  \\ 
 \hline
 \\ 
 $\Delta D_{t}$ &          1.00 &      0.03 &  \quad &          $\Delta P_{t}$ 	&      0.03 &  1.00 \\
\\
 $\Delta D_{t-1}$ &      0.22 &    -0.00 &     \quad  &     $\Delta P_{t-1}$ &      0.44 &  0.14 \\
 \\
 $\Delta D_{t-2}$ &     -0.16 &      -0.00 &     \quad &      $\Delta P_{t-2}$ &     0.06 &  -0.11 \\
\\
 $\Delta D_{t-3}$ &     -0.17 &     -0.06 &       \quad &    $\Delta P_{t-3}$ &     -0.19 &   -0.09 \\
\\
 $\Delta D_{t-4}$ &     -0.15 &    -0.08 &        \quad   &  $\Delta P_{t-4}$ &     -0.15 & -0.16 \\
 \\
 $\Delta D_{t-5}$ &     -0.11 &    -0.02 &     \quad  &      $\Delta P_{t-5}$ &    -0.06 &  -0.19 \\
\\
\hline
\\
\multicolumn{7}{c}{$\sigma(\Delta D_{t})$ =0.003, \, $\sigma(\Delta P_{t})$ = 0.246 \, $\sigma(\Delta D_{t})/\sigma(\Delta P_{t}) = 0.013$  } \\
\\
\hline \hline
\end{tabular}
\captionsetup{margin={1.5cm,1.5cm}, font=footnotesize, labelformat=empty, format=plain, justification=justified}
\caption{Table reports descriptive statistics of data used in the model estimation.}
\end{table*}

\newpage{}

\center Table VII. OLS Regression with heteroskedasticity-robust standard error.
\begin{table*}[h!]
\centering
\def\sym#1{\ifmmode^{#1}\else\(^{#1}\)\fi}
\begin{tabular}{l*{1}{c}}
\hline\hline
 Dependent Variable           &\multicolumn{1}{c}{dividend}\\
\hline
\\
price       &      0.0126\sym{***}\\
            &     (11.13)         \\
[1em]
\_cons      &      0.0730\sym{***}\\
            &     (18.81)         \\
            \\
\hline
\(N\)       &         140         \\
ADF-test residuals: -2.64*\\
Portmanteau (Q) statistic =   448.82\\
\hline
\multicolumn{2}{l}{\footnotesize \textit{t} statistics in parentheses}\\
\multicolumn{2}{l}{\footnotesize \sym{*} \(p<0.10\), \sym{**} \(p<0.05\), \sym{***} \(p<0.01\)}\\
\hline\hline
\end{tabular}
\captionsetup{margin={1.5cm,1.5cm}, font=footnotesize, labelformat=empty, format=plain, justification=justified}
\caption{An OLS regression is performed to derive the implied interest rate.}
\end{table*}

\newpage{}

\center Table VIII. Mapping the Maximum Log-Likelihood function 1871--2009, Equilibrium-Type A and Equilibrium-Type B.
\begin{table*}[h!]
\centering
\begin{tabular}{c c c c c }
\hline\hline
					& \multicolumn{2}{c}{Equilibrium-Type A} &\multicolumn{2}{c}{Equilibrium-Type B}\\ 
					\
 Interest rate assumption & NoNoise & Full Noise \qquad & NoNoise & Full Noise \\ 
 \hline
 \\
3\%    &      807.42   &  834.88    &  811.99     &  830.65  \\
\\

6\% &     810.82     &      833.15 &    808.49    &  831.47  \\
  \\
Free $r$   &    814.75    &      831.62 &   814.96      &    828.92     \\
        \hline 
\hline
\end{tabular}
\captionsetup{margin={1.5cm,1.5cm}, font=footnotesize, labelformat=empty, format=plain, justification=justified}
\caption{Maximum Likelihood for Equilibrium-Type A and Equilibrium-Type B during 1871--2009, with noise and without noise, and interest rate is 3\%, 6\% or estimated.}
\end{table*} 

\newpage{}

\center Table IX. Data 1871--2009: Price Coefficient of Equilibrium-Type A and Equilibrium-Type B.
\begin{table*}[h!]
\centering
\begin{tabular}{c c c c c c c c}		
\\	
\hline\hline
\\	
\multicolumn{8}{c}{Assumption: Interest rate = 0.03, Full Noise} \\ 
 \\
 \hline
   \bigskip 
 	 &	 &$p_{0} $ & $p_{D_{0}}$ & $ p_{I}$&$p_{D_{1}} $ & & ML	\\
  \bigskip 
 Equilibrium-Type A    &        &   -0.001   &   68.965    &  65.980 &  2.985  & & 834.00\\
 \bigskip
 Equilibrium-Type B    &        &   -0.061   &   39.809    & 39.803 &   0.007         &  & 830.65\\
  \hline 
  \\
\multicolumn{8}{c}{Assumption: Interest rate = 0.06, Full Noise} \\ 
\\
 \hline
   \bigskip 
 	 &	 &$p_{0} $ & $p_{D_{0}}$ & $ p_{I}$&$p_{D_{1}} $ & & ML	\\
  \bigskip 
 Equilibrium-Type A    &        &   -0.001   &  22.471 & 20.657 & 1.814  & & 833.15\\
 \bigskip
 Equilibrium-Type B    &        &   -0.006   &   73.246 & 73.236 & 0.004         &  & 831.47\\
    \hline 
\\
\multicolumn{8}{c}{Assumptions: free $r$ and free $\gamma$, Full Noise} \\ 
\\
 \hline
   \bigskip 
 	 &	 & $p_{0} $ & $p_{D_{0}}$ & $ p_{I}$&$p_{D_{1}} $ 	& & ML	\\
  \bigskip 
 Equilibrium-Type A$^{a}$    &        &   -0.000   &   64.888 & 52.342 & 12.546 & & 831.62			\\
 \bigskip
 Equilibrium-Type B$^{b}$    &        &   -0.016 &	4.963 &	4.892 & 0.071 &        & 828.92\\
\multicolumn{8}{l}{\footnotesize Note: a) interest rate estimated is 0.0309 }\\
\multicolumn{8}{l}{\footnotesize Note: b) interest rate estimated is 0.0686}\\
\hline\hline
\end{tabular}
\captionsetup{margin={1.5cm,1.5cm}, font=footnotesize, labelformat=empty, format=plain, justification=justified}
\caption{Price coefficient of models estimated in Table VIII with noise when interest rate is assumed 3\%, 6\% or estimated.}
\end{table*} 

\newpage{}

\center Table X. Time Series properties of Data 1995--2000
\begin{table*}[h!]
\centering
\begin{tabular}{c c c c c c }
\hline\hline
	 & \multicolumn{2}{c}{ADF} & \multicolumn{2}{c}{PP} & KPSS 	\\
	 	& drift		& trend		&	drift	&trend	& \\
\hline
	 \\
Price 	&	 -1.39	&	-1.43		&  -1.41  	&-1.17	&	0.14	\\
\\
Dividend 	& 	-2.51		&	0.37		&   -2.06	& 2.05	&	0.26	\\
\\
D.Price$^{*}$ 	&	-5.70		& -5.83  		&  	-7.06 & -7.17	&	0.15	\\
\\
D.Dividend &	-3.18	   	&	-4.54	  	&	-3.63 &  -4.72	&	0.17	\\
\\
\hline
\\
\multicolumn{6}{l}{Cointegrating regression: $D_{t}= 0.049+ 0.0017P_{t}+\epsilon_{t}$}\\
\multicolumn{6}{l}{Johansen test: rank$(r)=1$: $\lambda max =1.23$ (3.76 at 5\%) }\\
\multicolumn{6}{l}{\qquad \qquad \qquad \,	 rank$(r)=1$: $\lambda trace =1.23$ (3.76 at 5\%) } \\
\multicolumn{6}{l}{(*) D. is a differentiated variable by lag 1}\\
\\
\hline\hline
\end{tabular}
\captionsetup{margin={1.5cm,1.5cm}, font=footnotesize, labelformat=empty, format=plain, justification=justified}
\caption{Table reports descriptive statistics of data used in the model estimation.}
\end{table*} 

\newpage{}

\center Table XI. Mapping the Maximum Log-Likelihood function 1995--2000, Equilibrium-Type A and Equilibrium-Type B.
\begin{table*}[h!]
\centering
\begin{tabular}{c c c c }
\hline\hline
\\					
 Interest rate assumption  & Equilibrium-Type A \qquad  & Equilibrium-Type B  & LR-test\\ 
\\ \hline
 \\
 1.5\%    &         540.96	&   636.39  & 190.86$^{***}$ \\
\\
3\% &             487.17     	&    638.58  & 302.08$^{***}$ \\
 \\
  Free $r$  &            489.31 &        638.43     &  298.83$^{***}$ \\
 \\
        \hline 
\hline
\multicolumn{4}{l}{\footnotesize LR-test: -2 (Equilibrium-Type A-Equilibrium-Type B)}\\
\multicolumn{4}{l}{\footnotesize $\chi_{3}$ = 7.82 ($5\%=^{*}$), 11.35 ($1\%=^{**}$), 16.27 ($0.1\%=^{***})$ }\\
\end{tabular}
\captionsetup{margin={1.5cm,1.5cm}, font=footnotesize, labelformat=empty, format=plain, justification=justified}
\caption{Maximum Likelihood for Equilibrium-Type A and Equilibrium-Type B during 1995--2000, with noise and without noise, and interest rate is 3\%, 6\% or estimated.}
\end{table*}

\newpage{}

\center Table XII. Data 1995--2000: Price's Coefficient of Equilibrium-Type A and Equilibrium-Type B. 
\begin{table*}[h!]
\centering
\begin{tabular}{c c c c c c c c}		
\\
\hline \hline
\\
\multicolumn{8}{c}{Assumption: Interest rate = 1.5\%, Full Noise} \\ 
\\
 \hline
   \bigskip 
 	 &	 &$p_{0} $ & $p_{D_{0}}$ & $ p_{I}$&$p_{D_{1}} $ & & ML	\\
  \bigskip 
 Equilibrium-Type A    &        &    -0.177	&	71.942	& 67.882		& 4.060	& & 540.96\\
 \bigskip
 Equilibrium-Type B    &        &   -0.061	&    71.944	& 0.001		& 1.089     &  & 636.39\\
    \hline 
    \\
\multicolumn{8}{c}{Assumption: Interest rate = 3\%, Full Noise} \\
\\ 
 \hline
   \bigskip 
 	 &	 &$p_{0} $ & $p_{D_{0}}$ & $ p_{I}$&$p_{D_{1}} $ & & ML	\\
  \bigskip 
 Equilibrium-Type A    &        &  -0.124 	& 34.602 	& 31.147 		& 3.455 	&   & 487.17 \\
 \bigskip
 Equilibrium-Type B    &        & -0.430	&34.604 	& 0.002		&0.002     &  & 638.58 \\
    \hline
\\
\multicolumn{8}{c}{Assumptions: free $r$ and free $\gamma$, Full Noise} \\ 
\\
 \hline
   \bigskip 
 	 &	 & $p_{0} $ & $p_{D_{0}}$ & $ p_{I}$&$p_{D_{1}} $ 	& & ML	\\
  \bigskip 
 Equilibrium-Type A$^{a}$    &        &   -0.506	& 210.13	& 70.755	& 139.23 & & 489.31			\\
 \bigskip
 Equilibrium-Type B$^{b}$    &        &   -0.883	& 0.004	& 1089.196	& 118.025    &    & 638.43\\
\multicolumn{8}{l}{\footnotesize Note: a) interest rate estimated is 0.006 }\\
\multicolumn{8}{l}{\footnotesize Note: b) interest rate estimated is 0.002}\\
\hline\hline
\end{tabular}
\captionsetup{margin={1.5cm,1.5cm}, font=footnotesize, labelformat=empty, format=plain, justification=justified}
\caption{Price coefficient of models estimated in Table XI with noise when interest rate is assumed 3\%, 6\% or estimated.}
\end{table*}

\newpage{}

\center Table XIII (a). Data 1995--2000. Parameter Estimation Model A. 
\begin{table*}[h!]
\centering
\begin{tabular}{c c c c}	
\hline\hline
\\
\multicolumn{4}{c}{Assumption: Interest rate = 1.5\% and Full Noise.} \\
\multicolumn{4}{c}{Maximum Log-Likelihood = 540.96}\\ 
\\
\hline
   \bigskip 
 	$\tilde{p}_{0} $ 	& $\tilde{p}_{D_{0}}$ 	 & $\tilde{p}_{I}$		&$\tilde{p}_{D_{1}} $ \\
	  -0.177	&	71.942	& 67.882		&	4.060	\\
 \hline 
   \\
$\alpha_{D}$ = 0.232 	& 	$\alpha_{I}$= 0.232		&	$\alpha_{\Theta}$ = 0.001 &   	$\varphi=16.75$		\\
	(0.11) 			& 	(0.000) 				&					(0.000) &		\\
\\
$\sigma_{D}$= 0.001 	&	$ \sigma_{\Theta}$=0.001 &$\rho_{I}$=1.299			 & 	$\lambda$ = 0.002\\
	(0.02) 			& 	(0.001) 				&		(0.002)			 &		(0.031)\\
\\
\hline \hline
\end{tabular}
\captionsetup{margin={1.5cm,1.5cm}, font=footnotesize, labelformat=empty, format=plain, justification=justified}
\caption{Estimated parameters of Model A when interest rate is 1.5\% as reported in Table XII.}
\end{table*} 

\newpage{}

\center Table XIII (b). Data 1995--2000. Parameter Estimation Model B. 
\begin{table*}[h!]
\centering
\begin{tabular}{c c c c}	
\hline\hline
\\
\multicolumn{4}{c}{Assumption: Interest rate = 1.5\% and Full Noise.} \\
\multicolumn{4}{c}{Maximum Log-Likelihood = 636.39}\\ 
\\
\hline
   \bigskip 
 	$\widehat{p}_{0} $ 	& $\widehat{p}_{D_{0}}$ 	 & $ \widehat{p}_{I}$		&$\widehat{p}_{D_{1}} $ \\
	-0.061	&71.944		& 0.001		& 1.089 \\
	 \\ 

 \bigskip
 	$\tilde{p}_{0} $ 		& $\tilde{p}_{D_{0}}$ 	 & $\tilde{p}_{I}$		&$\tilde{p}_{D_{1}} $ \\
	-0.071	&71.942		& 70.854		& 67.114 \\
 \hline 
   \\
$\alpha_{D}$ = 0.001 	& 	$\alpha_{I}$= 0.905		&	$\alpha_{\Theta}$ = 0.001 &   			\\
	(0.11) 			& 	(0.000) 				&					(0.000) &		\\
\\
$\sigma_{D}$= 0.001 	&	$ \sigma_{\Theta}$=0.124 &$\rho_{I}$=0.001			 & 	$\lambda$ = 0.001\\
	(0.004) 			& 	(0.001) 				&		(0.002)			 &		(0.001)\\
\\
\hline \hline
\end{tabular}
\captionsetup{margin={1.5cm,1.5cm}, font=footnotesize, labelformat=empty, format=plain, justification=justified}
\caption{Estimated parameters of Model B when interest rate is 1.5\% as reported in Table XII.}
\end{table*} 

\newpage{}

\end{document}